\documentclass[journal]{IEEEtran} 
\usepackage{graphicx}
\usepackage[cmex10]{amsmath}
 \usepackage{amsthm} 
\usepackage{amsfonts}
\usepackage{amssymb} 
\usepackage{graphicx}%
\usepackage{color}%
\usepackage{algorithmic}
 \usepackage{algorithm}
\newtheorem{theorem}{Theorem}

\usepackage[colorlinks=true, linkcolor=blue, citecolor=red]{hyperref}
\newtheorem{lemma}{Lemma}
\usepackage[english]{babel}
\usepackage{blindtext}
\usepackage{epsfig}
\usepackage{framed}
\usepackage{xcolor}
\newtheorem{remark}{Remark}
\newcommand{\supp}{\mathop{\mathrm{supp}}}
\definecolor{rick}{RGB}{0,139,0}

\IEEEoverridecommandlockouts

\title{On the Gaussian Multiple Access Wiretap Channel and the Gaussian Wiretap Channel with a Helper: Achievable Schemes and Upper Bounds}

\author{
\IEEEauthorblockN{Rick Fritschek and Gerhard Wunder\\}
\IEEEauthorblockA{Heisenberg Communications and Information Theory Group\\
    Freie Universit\"at Berlin, \\
    Takustr. 9,
    D--14195 Berlin, Germany\\
    Email: rick.fritschek@fu-berlin.de, g.wunder@fu-berlin.de
    \thanks{This paper was presented in part at the ISIT 2016 \cite{FW16}, and in the preprint \cite{FW17axvMACWT}. }}
}

\begin{document}

\maketitle
\begin{abstract}
We study deterministic approximations of the Gaussian two-user multiple access wiretap channel (G-MAC-WT) and the Gaussian wiretap channel with a helper (G-WT-H). These approximations enable results beyond the recently shown 2/3 and 1/2 secure degrees of freedom (s.d.o.f.) for the G-MAC-WT and the G-WT-H, respectively. While the s.d.o.f. were obtained by real interference alignment, our approach uses signal-scale alignment. We show achievable schemes which are independent of the rationality of the channel gains. Moreover, our results can differentiate between channel strengths, in particular between both users, and will establish secrecy rates dependent on this difference. We can show that the resulting achievable secrecy rates tend to the s.d.o.f. for vanishing channel gain differences. Moreover, we extend previous and develop new techniques to prove generalized s.d.o.f. bounds for varying channel strengths and show that our achievable schemes reach the bounds for certain channel gain parameters. We believe that our analysis is the next step towards a constant-gap analysis of the G-MAC-WT and the G-WT-H. 
\end{abstract}

\section{Introduction}
The wiretap channel was first proposed by Wyner in \cite{Wyner75}, and solved in its degraded version. This result was later extended to the general wiretap channel by Csiszar and K\"orner in \cite{CsizarKoerner}. Moreover, the Gaussian equivalent was studied by Leung-Yan-Cheon and Hellman in \cite{Hellman}. The wiretap channel and its modified version served as an archetypical channel for physical-layer security investigations. However, in recent years, the support of multiple users became increasingly important. A straightforward extension of the wiretap channel to multiple users was done in \cite{TekinYenerGMAC-WT}, where the Gaussian multiple access wiretap channel (G-MAC-WT) was introduced. A general solution for the secure capacity of this multi-user wiretap set-up was out of reach and investigations focused on the secure degrees of freedom (s.d.of.) of these networks. Degrees of freedom are used to gain insights into the scaling behaviour of multi-user channels in comparison to a single-link scenario. They measure the capacity of the network, normalized by the single-link power, as power goes to infinity. This also means that the d.o.f. provide an asymptotic view on the problem at hand. This simplifies the analysis and enables asymptotic solutions of channel models where no finite power capacity results could be found. A disadvantage of the d.o.f. is that they do not incorporate channel gain differences, as the limiting process is only on the the signal power itself. This limits the insights from the d.o.f. about the underlying capacity region, as those channel gain differences usually play an important role in multi-user scenarios, for example in the Gaussian interference channel or the Gaussian multiple access channel. An example for a technique which yields d.o.f. achievability results is real interference alignment \cite{motahari2014real}. It uses integer lattice transmit constellations which are scaled such that alignment can be achieved. The intended messages are recovered by minimum-distance decoding and the minimum distance is analysed and bounded by usage of the Khintchine-Groshev theorem of Diophantine approximation theory~\cite[Appendix A]{motahari2014real}. The disadvantage of the method is that these results only hold for almost all channel gains. This is unsatisfying for secrecy purposes since it leaves an infinite amount of cases where the schemes do not work, e.g. rational channel gains. Moreover, secrecy should not depend on the accuracy of channel measurements. Real interference alignment is part of a broader class of interference alignment strategies. Interference alignment (IA) was introduced in \cite{CadambeJafarIA} and \cite{Maddah-AliKhandi-IA}, among others, and its main idea is to design signals such that the interference overlaps (aligns) and therefore uses fewer signal dimensions. The resulting interference-free signal dimensions can be used for communication. IA methods can be divided into two categories, namely the vector-space alignment approach and the signal-scale alignment approach \cite{Niesen-Ali}. 
The former uses the classical signalling dimensions of time, frequency and multiple-antennas for the alignment, while the latter uses the signal strength for alignment. Real interference alignment and signal-strength deterministic models are examples for signal-scale alignment. Signal-strength deterministic models are based on an approximation of the Gaussian channel. An example for such an approximation is the linear deterministic model (LDM), introduced by Avestimehr et al. in \cite{Avestimehr2011}. It is based on a binary expansion of the transmit signal, and an approximation of the channel gain to powers of two. The resulting binary expansion gets truncated at the noise level which yields a noise-free binary signal vector and makes the model deterministic. The next step towards constant-gap capacity results are the {\it generalized} d.o.f., first investigated in \cite{Etkin2008}, which study the limiting process of the signal-to-noise ratio to infinity. The difference to d.o.f. results is, that channel gain differences are taken into account. The g.d.o.f. can therefore add valuable insights about the structure of the underlying capacity region. The linear deterministic model yields results which directly correspond to the g.d.o.f. of various Gaussian channels (i.e. \cite{Bresler2008}, \cite{Bresler2010}, \cite{Suvarup2011}, \cite{FW17b}). Moreover, in those examples the capacity can be approximated by the LDM such that it is within a constant bit-gap of the Gaussian channel. This is due to layered lattice coding schemes which can be used to transfer the achievable scheme to the Gaussian model. Moreover, converse proofs for the deterministic models can be often translated to the Gaussian case and are therefore helpful for constant-gap results.

{\bf Previous and related work:} 
Previous work on the wiretap channel in multi-user settings mainly utilized the real IA approach in addition to cooperative jamming, introduced in \cite{TekinYenerCoopJam}. The idea of using IA in a secrecy context is to cooperatively jam the eavesdropper, while aligning the jamming signal in a small subspace at the legitimate receiver. This resulted in a sum s.d.o.f characterization of $\tfrac{K(K-1)}{K(K-1)+1}$ for the $K$-user case of the G-MAC-WT in \cite{XieUlukusOneHop}. 
The idea is that the users can allocate a small part of the signalling dimensions with uniformly distributed random bits. Those random bits are send such that they occupy a small space at the legitimate receiver, while overlapping with the signals at the eavesdropper. A specialized model is the Gaussian wiretap channel with a helper (G-WT-H). This model consists of the standard wiretap channel model, with a second independent user, whose only purpose is to jam the eavesdropper. In \cite{XieUlukusWiretap-Helper} and \cite{XieUlukusWiretap-Helper2}, the real IA approach was used on the G-WT-H (with and without channel state information, respectively) to investigate the s.d.o.f, therefore achieve results for the infinite power regime. They showed that the sum s.d.o.f is $\tfrac{1}{2}$ for the $2$-user case. Another branch of recent work \cite{Comp_forward_mac_wt} approached the problem, using a compute-and-forward decoding strategy, which leads to results for the finite power regime that are optimal in an s.d.o.f sense. The next step is to transition from the s.d.o.f. results, to a secure constant-gap capacity result. A promising approach is to study linear deterministic approximations to gain insights leading to constant-gap capacity approximations. This approach has been used for example for wiretap channels in \cite{Shamai12,VogtSezgin14}, for relay networks \cite{PeronDiggaviTelatar09} and for IC channels \cite{MM13}, \cite{Vogt16}.  It was also recently used in \cite{LeeKhisti17} for an s.d.o.f. analysis of the Gaussian diamond-wiretap channel, which is a multi-hop version of the G-MAC-WT. 

{\bf Contributions:}
We follow the deterministic approximation approach and investigate the G-WTH-H and G-MAC-WT models. For that we assume perfect knowledge of the magnitude of all channel gains in the network. We show achievable rate results for the G-WT-H for general channel gain strengths and a finite SNR regime, independent of the channel gain being rational or not. Moreover, we develop a converse proof which shows a constant-gap for certain channel gain parameter ranges. The upper bound converges to the s.d.o.f bound for vanishing channel gain differences. These results were already shown in the conference version of this publication \cite{FW16}. The present work extends those results by providing another more elegant achievability proof and the full converse proof of the G-WT-H without the assumption of an uniform input distribution. Both of our results on the G-WT-H combined give insights into the secure g.d.o.f. (see. Section~\ref{sgdof}) and provide tools and upper bounds for future constant-gap capacity results.  Moreover, we use the same alignment methods to show an achievable scheme for the linear deterministic MAC-WT (LD-MAC-WT) which is dependent on channel gain differences and therefore gives insights into the secure g.d.o.f. as well. We show that both achievable schemes can be translated to the Gaussian channel models, by using layered lattice codes to imitate bit-levels. We extend the converse proof of \cite{XieUlukusOneHop} for the G-MAC-WT d.o.f towards general channel gain strengths, to match our achievable scheme for certain channel parameter ranges. For that, we combine previous techniques with new novel techniques to translate the results of both converse proofs to the Gaussian channel.
\begin{figure}
\centering
\includegraphics[scale=0.82]{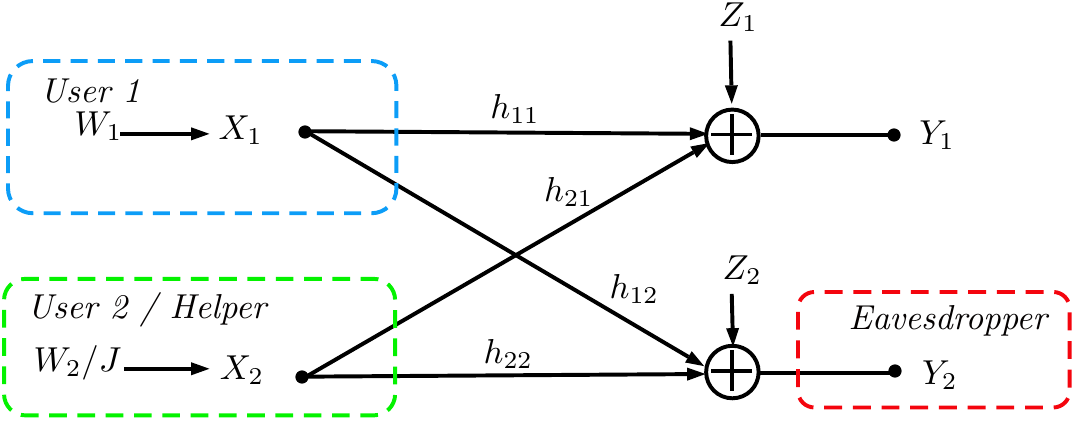}
\caption{The Gaussian multiple access wiretap channel and the Gaussian wiretap channel with one helper. The difference between both models is, that in the former, user 2 wants to send a message ($W_2$), while in the later model, he only helps user 1 by jamming with $J$.}
\label{System_Model_Figure}
\end{figure}
\label{System model}
\section{System Model}
The Gaussian multiple access wiretap channel (G-MAC-WT) and the Gaussian wiretap channel with one helper (G-WT-H) are defined as a system consisting of 2 transmitters and 2 receivers, as shown in Fig.~\ref{System_Model_Figure}, where $X_1,X_2 \in \mathbb{R}$ are the channel inputs of both users, communicating with the legitimate receiver with channel output $Y_1$ or jamming the eavesdropper, with channel output $Y_2$.
The channel itself is modelled with additive white Gaussian noise, $Z_1,Z_2\sim \mathcal{N}(0,1)$. Therefore, the system equations can be written as 
\begin{IEEEeqnarray}{rCl}\label{Gauss_Model_Helper}
Y_1&=&h_{11}X_1+h_{21}X_2+Z_1\IEEEyessubnumber,\\
Y_2&=&h_{22}X_2+h_{12}X_1+Z_2\IEEEyessubnumber,
\end{IEEEeqnarray}
where the channel inputs satisfy an average transmit power constraint $E\{X_i^2 \} \leq P_i$ for each $i$. The channel gains from user $i$ to receiver $k$ are denoted by $h_{ik}$. Let $|h_{11}|^2P_1=\text{SNR}_1$ and $|h_{21}|^2P_2=\text{SNR}_2$ represent the received average power at $Y_1$ of both direct signals. We assume that both signals are received at $Y_2$ with the same average power and therefore $h_{12}=h_{22}=h_E$ and $P_1=P_2=P$ which gives $|h_E|^2P=\text{SNR}_E$.\footnote{This will reduce the number of cases and therefore simplify the analysis. However, the following techniques also work without this assumption. See  also remark 3.} We introduce the two parameters $\beta_1$ and $\beta_2$, which connect the $\text{SNR}$ ratios with $\text{SNR}_2=\text{SNR}_1^{\beta_1}$ and $\text{SNR}_E=\text{SNR}_1^{\beta_2}$.  The difference between the G-MAC-WT and G-WT-H is, that in case of the G-WT-H, user $2$ is just helping user $1$ by independently jamming both receivers to achieve a secure communication. In the case of the G-MAC-WT model, both users want to transmit information to $Y_1$ and are able to use jamming.

\subsubsection{G-WT-H}
A $(2^{nR},n)$ code will consist of an encoding and a decoding function. The encoder assigns a codeword $x_1^n(w)$ to each message $w$, where $W$ is uniformly distributed over the set $[1:2^{nR}]$, and the associated decoder assigns an estimate $\hat{w}\in [1:2^{nR}]$ to each observation of $Y_1^n$. A secure rate $R$ is said to be achievable if there exist a sequence of $(2^{nR},n)$ codes which satisfy a probability of error constraint $P^{(n)}_e\!=P\!(\hat{W}\!\neq\! W)\leq\epsilon$ as well as a secrecy constraint 
\begin{equation}
\tfrac{1}{n} H(W|Y_2^n) \geq \tfrac{1}{n} H(W)-\epsilon,
\label{security_constraint}
\end{equation}
which gives $I(W;Y_2^n)\leq \epsilon n$ where $\epsilon\rightarrow 0$ for $n\rightarrow \infty$.
A message $W$ is therefore information-theoretically secure if the eavesdropper cannot reconstruct it from the channel observation $Y_2^n$. This means that the uncertainty of the message is almost equal to its entropy, given the channel observation.
\subsubsection{G-MAC-WT}
A $(2^{nR_1},2^{nR_2},n)$ code for the multiple access wiretap channel will consist of a message pair $(W_1,W_2)$ uniformly distributed over the message set $[1:2^{nR_1}]\times[1:2^{nR_2}]$ with a decoding and two randomized encoding functions. Encoder $1$ assigns a codeword $X_1^n(w_1)$ to each message $w_1\in [1:2^{nR_1}]$, while the encoder $2$ assigns a codeword $X_2^n(w_2)$ to each message $w_2\in [1:2^{nR_2}]$. The decoder assigns  an estimate $(\hat{w_1},\hat{w_2})\in [1:2^{nR_1}]\times [1:2^{nR_2}]$ to each observation of $Y_1^n$. A secure rate pair ($R_1,R_2$) is said to be achievable if there exist a sequence of $(2^{nR_1},2^{nR_2},n)$ codes, which satisfy a reliability constraint, i.e. probability of error such that: $P^{(n)}_e=P[(\hat{W_1},\hat{W_2})\neq (W_1,W_2)]\leq\epsilon$  and a security constraint for both messages $W_1,W_2$:
\begin{equation}
\tfrac{1}{n} H(W_1,W_2|Y_2^n) \geq \tfrac{1}{n} H(W_1,W_2)-\epsilon,
\label{security_constraint2}
\end{equation}
which gives $I(W_1,W_2;Y_2^n)\leq \epsilon n$, where $\epsilon\rightarrow 0$ for $n\rightarrow \infty$. In particular, for the G-MAC-WT model, we are interested in the secure {\it sum}-rate $R_\Sigma:=R_1+R_2$.

\subsection{Secure Generalized Degrees of Freedom}
\label{sgdof}
The following example is from \cite[Section 2]{Bresler2008} and we encourage the reader to look into that publication for a more in-depth discussion of degrees-of-freedom. Lets have a look at the motivating example of the multiple access channel. The channel model is defined by
\begin{equation*}
Y=h_1X_1+h_2X_2+Z,
\end{equation*}
where the channel gains are $h_1,h_2\in \mathbb{C}$ and the input signals have the power constraint $E\{|X_i|^2\}\leq P$. Moreover, the channel has an additive Gaussian noise $Z\sim \mathcal{CN}(0,1)$. The model is parametrized by the signal-to-noise ratios $\text{SNR}_1:=|h_1|^2P$ and $\text{SNR}_2:=|h_2|^2P$. The capacity region of this channel model is
\begin{IEEEeqnarray*}{rCl}
R_1 &\leq& \log(1+\text{SNR}_1)\approx \log \text{SNR}_1,\\
R_2 &\leq& \log(1+\text{SNR}_2)\approx \log\text{SNR}_2,\\
R_1+R_2 &\leq& \log(1+\text{SNR}_1+\text{SNR}_2)\approx \text{SNR}_1,
\end{IEEEeqnarray*} where we assumed that $\text{SNR}_1>\text{SNR}_2$. The degrees-of-freedom are a way to simplify the analysis by looking into the scaling behaviour of multi-user channels in comparison to a single-link channel. They are defined as follows
\begin{equation}
D:= \lim_{P\rightarrow\infty} \frac{C(h_1,h_2,P)}{\log P}.
\label{dof}
\end{equation}
The reason behind this formula is that the single link capacity of the AWGN channel is 
\begin{equation}
C(\text{SNR})=\log (1+\text{SNR})\approx \log \text{SNR},
\end{equation} and the d.o.f therefore provide a scaling of the multi-user channel in comparison to the single-link capacity, where power goes to infinity. The d.o.f for the MAC are therefore \begin{equation*}
d_1\leq 1,\qquad d_2\leq 1 \qquad \text{and} \qquad d_1+d_2\leq 1.
\end{equation*}
One can see that the d.o.f are independent of the channel gains and do not reflect subtleties in the capacity region. In \cite{Etkin2008}, a more sophisticated approach, coined generalized d.o.f, was suggested. There, the ratio between both signal-to-noise ratios is a fixed constant such that $\text{SNR}_1^\alpha=\text{SNR}_2$. The g.d.o.f. are then defined as
\begin{equation*}
D(\alpha):= \lim_{\text{SNR}_1 \rightarrow \infty} \frac{C(\text{SNR}_1,\text{SNR}_1^\alpha)}{\log \text{SNR}_1}.
\end{equation*}
Notice that the scale difference of both received signals directly influences the generalized d.o.f. Plugging in the capacity region of the MAC results in 
\begin{equation*}
d_1\leq 1,\qquad d_2\leq\alpha, \qquad\text{and} \qquad d_1+d_2\leq 1,
\end{equation*} which preserves the subtleties of the finite SNR capacity region. The g.d.o.f. therefore provide valuable insights on the behaviour of the capacity region as a function of the channel gain differences. Now, the secure d.o.f. are defined as in \eqref{dof}, with the secure capacity instead, which is the supremum of all achievable secrecy rates. For example, the s.d.o.f. of the G-WT-H are $D_s\leq \tfrac{1}{2}$. Now, it could be the case that most of the signal communication between the user and the legitimate receiver is secure simply because that part vanishes under the noise floor at the eavesdropper, which is not reflected in the s.d.o.f. It is therefore important to get insights into the capacity behaviour in dependence on channel gains with the secure g.d.o.f. The secure g.d.o.f. for a complex valued channel model are defined as
\begin{equation}
D_s(\beta_1,\beta_2):= \lim_{\text{SNR}_1\rightarrow \infty} \frac{C_s(\text{SNR}_1,\text{SNR}_1^{\beta_1},\text{SNR}_1^{\beta_2})}{\log \text{SNR}_1}.
\end{equation} We will look into real valued channels, which changes the scaling to $\tfrac{1}{2}\log \text{SNR}_1$.

\subsection{The Linear Deterministic Approximation}

As simplification, we will investigate the corresponding linear deterministic model (LDM)\cite{Avestimehr2011} of the system models as an intermediate step. For the LDM, inputs are assumed to have a unit average power constraint, while the channel gains represent the SNR, such that $h=\sqrt{\text{SNR}}\approx 2^n$. The LDM models the signals of the channel as bit-vectors $\mathbf{X}$ or equivalently as a succession of bits in a scalar $x=0.b_1b_2b_3\ldots$, which is achieved by a binary expansion of the real-valued input signal $X$. The positions within the bit-vector are referred to as bit-levels. The channel gain $2^n$ shifts the bits of a scalar input for $n$-positions over the decimal point, such that we have $2^nx=b_1b_2\ldots b_n. b_{n+1}b_{n+2}\ldots$ The noise only affects the bits on the right hand side of the decimal point $2^nx+z=b_1b_2\ldots b_n. \underbar{$b$}_{n+1}\underbar{$b$}_{n+2}\ldots$ denoted as $\underbar{$b$}_{n+i}$. The deterministic approximation cuts of the noise effected bits after the shift, which results in $y\approx b_1b_2\ldots b_n$. This truncation at the lowest level (noise level) models the signal impairment of the Gaussian noise, which yields a deterministic approximation of the Gaussian model. Viewing this in the equivalent algebraic notation, channel gains are represented by shifting the input bit-vector for an appropriate number of bit-levels down. The shift is introduced by a $q\times q$ shift-matrix $\mathbf{S}$, which is defined as
\begin{equation}
\mathbf{S}=\begin{pmatrix}
0 & 0 &  \cdots & 0 & 0\\
1 & 0 &  \cdots & 0 & 0\\
0 & 1 &  \cdots & 0 & 0\\
\vdots & \vdots & \ddots & \vdots & \vdots \\
0 & 0 &  \cdots & 1 & 0\\
\end{pmatrix}.
\end{equation}
With $\mathbf{S}$, an incoming bit vector can be shifted for $q-n$ positions with $\mathbf{Y}=\mathbf{S}^{q-n}\mathbf{X}$. Therefore, the channel gain is now represented by $n$ bit-levels which corresponds to $\lfloor\tfrac{1}{2}\log \text{SNR}\rfloor$ of the original channel. Furthermore, superposition of different signals is modelled by binary addition on the bit-levels, i.e. element-wise addition modulo-two. Carry over is not used to limit the superposition on the specific level where it occurs.
With these definitions, the model can be written as
\begin{IEEEeqnarray}{rCl}
\mathbf{Y}_1&=&\mathbf{S}^{q-n_{11}}\mathbf{X}'_1\oplus\mathbf{S}^{q-n_{21}}\mathbf{X}'_2\IEEEyessubnumber,\\
\mathbf{Y}_2&=&\mathbf{S}^{q-n_{22}}\mathbf{X}'_2\oplus\mathbf{S}^{q-n_{12}}\mathbf{X}'_1\IEEEyessubnumber,
\label{LDM_Model}
\end{IEEEeqnarray}
where $q:=\max\{n_{11},n_{12},n_{21},n_{22}\}$. For ease of notation, we denote $\mathbf{X}_1=\mathbf{S}^{q-n_{11}}\mathbf{X}'_1$ and $\mathbf{X}_2=\mathbf{S}^{q-n_{21}}\mathbf{X}'_2$.
Furthermore, we denote $\mathbf{S}^{q-n_{22}}\mathbf{X}'_2$ and $\mathbf{S}^{q-n_{12}}\mathbf{X}'_1$ by $\bar{\mathbf{X}}_2$ and $\bar{\mathbf{X}}_1$, respectively. We also include the assumption on the symmetry in the channel gains at the eavesdropper, which leads to $n_{22}=n_{12}=:n_E$, and denote $|n_{1}-n_{2}|=:n_\Delta$ with $n_{11}=:n_{1}$ and $n_{21}=:n_2$. 
We can therefore rewrite the deterministic channel model as
\begin{IEEEeqnarray}{rCl}
\mathbf{Y}_1&=&\mathbf{S}^{q-n_{1}}\mathbf{X}'_1\oplus\mathbf{S}^{q-n_{2}}\mathbf{X}'_2=\mathbf{X}_1\oplus\mathbf{X}_2\IEEEyessubnumber,\\
\mathbf{Y}_2&=&\mathbf{S}^{q-n_{E}}\mathbf{X}'_2\oplus\mathbf{S}^{q-n_{E}}\mathbf{X}'_1=\mathbf{\bar{X}}_2\oplus\mathbf{\bar{X}}_1\IEEEyessubnumber.
\label{LDM_Model}
\end{IEEEeqnarray}
The resulting received bit-vectors of the channel model can be illustrated as shown in Fig.~\ref{wiretap-LDM}. There, one can see that for example the two bit-vectors $\mathbf{X}_1$ and $\mathbf{X}_2$ are received at $\mathbf{Y}_1$ with $n_1$ and $n_2$ bit-levels, respectively. The highest bit is at the top of the boxes, while the lowest bit is just above the noise level.
All schemes rely on a partition of the received signal of the legitimate receiver into a common ($\mathbf{Y}_{1,c}$) and a private ($\mathbf{Y}_{1,p}$) part. The common bits are the top 
\begin{equation}
n_c:=\min\{n_E+n_\Delta,\max\{n_{1},n_{2}\}\}
\label{common_part_alignment_scheme}
\end{equation} bits of $\mathbf{Y}_1$. And the private part consists of the bottom 
\begin{equation}
n_p:=(\max\{n_{1},n_{2}\}-n_c)^+
\end{equation} bits of $\mathbf{Y}_1$. We note that due to the bit-level shift, the last $n_\Delta$ bits of $\mathbf{X}_1$ in $\mathbf{Y}_{1,c}$ are actually private, see Remark~\ref{remark1} and Fig.~\ref{wiretap-LDM}.
To specify a particular range of elements in a bit-level vector we use the notation $\mathbf{a}_{[i:j]}$ to indicate that $\mathbf{a}$ is restricted to the bit-levels $i$ to $j$. Bit-levels are counted from top, most significant bit in the expansion, to bottom. If $i=1$, it will be omitted $\mathbf{a}_{[:j]}$, the same for $j\!=\!n$ $\mathbf{a}_{[i:]}$. We define the modulo operation as $a\bmod n:=a-\lfloor \tfrac{a}{n}\rfloor n$.

\begin{remark}
The assumption that $n_{22}=n_{12}=n_E$, i.e. the eavesdropper receives the signals with equal strength, does not influence the achievable secrecy sum-rate of the LD system. Consider a channel with $n_{22}\neq n_{12}$, for example $n_{22}>n_{12}$. The part of $\mathbf{X}'_2$ which is received above $n_{12}$ at the eavesdropper, $\bar{\mathbf{X}}_{2,[:n_{22}-n_{12}]}$, cannot be utilized since it cannot be jammed. One can therefore achieve the same rate by ignoring the top $n_{22}-n_{12}$ bits of $\mathbf{X}'_2$. The same argument holds for $n_{12}>n_{22}$.
\end{remark}

\section{Main Results}
\begin{figure}
\centering
\includegraphics[scale=0.9]{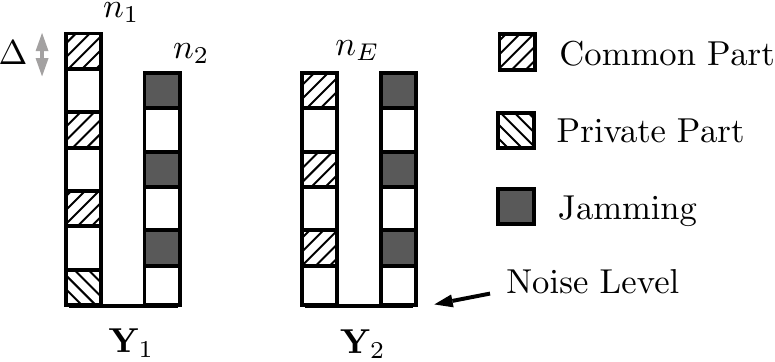}
\caption{The Gaussian wiretap channel with a helper in the linear deterministic model. The helper utilizes jamming, such that all used signal parts align with the jamming at the attacker ($Y_2$).}
\label{wiretap-LDM}
\end{figure}
\subsection{Results for the LD-WT-H}
We start with the achievable sum-rate of the linear deterministic wiretap channel with a helper. We need to differentiate between the cases $n_{1}\geq n_{2}$ and $n_{1}< n_{2}$, since the user and the helper have different roles and those cases lead to different schemes. The main idea is to deploy a jamming scheme such that the jamming signal parts of the helper overlap (align) with the used signal parts of the user, at $\mathbf{Y}_2$, while minimizing the overlap at the legitimate receiver $\mathbf{Y}_1$, see Fig.~\ref{wiretap-LDM}. This leads to the following result:
\begin{theorem}
\label{theorem_det}
An achievable secrecy rate $R$ of the linear deterministic wiretap channel with a helper is
\begin{equation}
R=n_p+\left\lfloor \frac{n_c}{2n_\Delta}\right\rfloor n_\Delta+Q
\end{equation}
where 
\begin{equation}\label{reminder_ld_helper}
Q=\begin{cases}
n_Q & \text{for } n_Q< n_\Delta, n_1\geq n_2
\\
n_\Delta & \text{for } n_Q \geq  n_\Delta, n_1\geq n_2\\
0 & \text{for } n_Q< n_\Delta, n_1< n_2\\
n_Q-n_\Delta & \text{for } n_Q \geq  n_\Delta, n_1< n_2
\end{cases}
\end{equation} with $n_Q:=n_c \bmod 2n_\Delta$.
\end{theorem}
The proof is presented in the analysis section \ref{ach-scheme-det-model}.
This theorem says, that the secure rate is the sum of the private bits $n_p$, which are just received by the legitimate receiver, due to channel gain differences, plus half of the common bits $n_c$ and an additional term $Q$ accounting for the remainder bits of the partitioning scheme. We now present an upper bound to the sum-rate.
\begin{theorem}\label{THM-converse-LD-WT-Helper}
The secrecy rate $R$ of the linear deterministic wiretap channel with one helper and symmetric channel gains at the wiretapper is bounded from above by
\begin{equation*}
R \leq \min\{ r_{ub1}, r_{ub2}, r_{ub3}\}
\end{equation*}
with \begin{IEEEeqnarray*}{rCl}
r_{ub1}&=& n_p+\tfrac{1}{2}n_c+\tfrac{1}{2}(n_1-n_2)^+\\
r_{ub2}&=&n_{1}\\
r_{ub3}&=&n_{2}+(n_{1}-n_{2}-n_{E})^+\\
&& +\:[n_E-n_{2}-(n_E-n_{1}+n_{2})^+]^+
\end{IEEEeqnarray*}
\end{theorem}
The proof was first presented in \cite{FW16} and is in the same fashion as for the truncated deterministic model, which is part of the analysis of the Gaussian model, see Section~\ref{LD-Converse} for further details. 

\subsection{Results for the LD-MAC-WT}
The main idea stays the same: we use a scheme which aligns the jamming parts with the signal parts at the adversary receiver, while minimizing the alignment at the legitimate receiver. But there is the important difference that the channel is symmetrical, i.e. both users can send messages and jam which enables a higher achievable rate. We can also assume w.l.o.g that $n_1\geq n_2$, due to this symmetry. Our analysis in section~\ref{Ach_Scheme_LD_MAC_WT} shows the following:
\begin{theorem}
An achievable secrecy sum-rate $R_{\Sigma}$ of the linear deterministic multiple access wiretap channel with symmetric channel gains at the eavesdropper is
\begin{equation}
R_{\Sigma}=
n_p+\lfloor \tfrac{n_c}{3n_\Delta} \rfloor 2n_\Delta+Q
\end{equation}
 and 
\begin{equation}
Q=\begin{cases}
q & \text{for } n_Q< n_\Delta 
\\
n_\Delta & \text{for } 2n_\Delta > n_Q \geq  n_\Delta\\
n_\Delta + q & \text{for } n_Q \geq 2n_\Delta,
\end{cases}
\end{equation} with $n_Q= n_c \bmod 3n_\Delta$ and $q=n_Q \bmod n_\Delta$.
\end{theorem}
Notice that, in the LD-WT-H, the common rate part was approximately $\tfrac{n_c}{2}$, but now we can achieve $\tfrac{2n_c}{3}$. We therefore have an increase in the common (secrecy) rate in comparison to the LD-WT-H. This is due to the aforementioned symmetry. This observation is in accordance with previous s.d.o.f results, as they are $\tfrac{1}{2}$ and $\tfrac{2}{3}$, for the WT-H and the MAC-WT, respectively. The next theorem provides the corresponding upper bound.

\begin{theorem}\label{THM-converse-LD-MAC-WT}
The secrecy sum-rate $R_{\Sigma}$ of the linear deterministic multiple access wiretap channel with symmetric channel gains at the eavesdropper is bounded from above by
\begin{equation}
R_{\Sigma}\leq\begin{cases}
n_p+\tfrac{2}{3}n_c+\frac{1}{3}n_\Delta & \text{for } n_2\geq n_E
\\
\tfrac{2}{3}n_c+\frac{1}{3}n_\Delta  & \text{for } n_E> n_2.\\
\end{cases}
\end{equation}
\end{theorem}
The proof is in the same fashion as the one for the truncated deterministic model, see Section~\ref{LD-Converse} for details.

\subsection{Results for the G-WT-H}
To get achievability results for the Gaussian wiretap channel with a helper, we stick to the previously developed scheme for the deterministic model. We will transfer the alignment and jamming structure of the previous model to its Gaussian equivalent with layered lattice codes, similar to \cite{Bresler2010}. This will lead to an achievable rate which is directly based on the deterministic rate. Moreover, we will make use of results in \cite{Mukherjee2017} to show that the mutual information of the Gaussian case can be upper bounded by an appropriate deterministic model. As a result, the bound for the deterministic model is also a bound for the Gaussian model, with a constant bit-gap.

\begin{theorem}
An achievable secrecy rate $R$ of the Gaussian Wiretap channel with a helper is
\begin{equation*}
R=r^p+r^c+r^R
\end{equation*}
where $r^c\!:=\!l_{\text{u}}(\tfrac{1}{2}\log \text{SNR}_1^{(1-\beta_1)}\!-\!\tfrac{1}{2})$, with\\ 
\begin{equation*}
l_{\text{u}}:= \left\lfloor \frac{\min\{1+\beta_2-\beta_1,1\}}{2(1-\beta_1)}\right\rfloor,
\end{equation*}\\
 $r^p:=\tfrac{1}{2}\log(\max\{1,\text{SNR}_1^{\beta_1-\beta_2}\})$, and\\
\begin{equation*}
r^R=\begin{cases}
r^{R_1} & \text{for } r^{R_1} < r^{R_2}, \text{SNR}_1\geq \text{SNR}_2
\\
r^{R_2} & \text{for } r^{R_1} \geq  r^{R_2}, \text{SNR}_1\geq \text{SNR}_2\\
0 & \text{for } r^{R_1} < r^{R_2}, \text{SNR}_2\geq \text{SNR}_1\\
r^{R_3} & \text{for } r^{R_1} \geq  r^{R_2}, \text{SNR}_2\geq \text{SNR}_1
\end{cases}
\end{equation*}
with
\begin{IEEEeqnarray*}{rCl}
r^{R_1}&:=& \tfrac{1}{2}\log \text{SNR}_1^{1-2l_\text{u}(1-\beta_1)} - \tfrac{1}{2}\log \text{SNR}_1^{\min\{ \beta_1-\beta_2,0\}}-\tfrac{1}{2},\\
r^{R_2}&:=& \tfrac{1}{2}\log \text{SNR}_1^{(1-\beta_1)}-\tfrac{1}{2},\quad r^{R_3}:= r^{R_1}-r^{R_2}.
\end{IEEEeqnarray*}
\end{theorem}
We note that this results directly corresponds to the linear deterministic result in Theorem~\ref{theorem_det}. For example $l_u$ corresponds to the signal scale factor $\lfloor\tfrac{n_c}{2n_\Delta}\rfloor$ and $\log\text{SNR}^{(1-\beta_1)}$ to $n_\Delta$, which can be seen by using \eqref{beta_1_corr}. We are therefore within a constant bit-gap of the rates of the deterministic model (Theorem~\ref{theorem_det}), by comparing via $n=\lfloor \tfrac{1}{2}\log \text{SNR} \rfloor$.
For the converse, the goal is to bound the Gaussian mutual information terms by the ones of the deterministic model. Due to the G-WT-H consisting of MAC channels with security constraint, one could try to use the constant-gap bound of \cite{Bresler2008}. Unfortunately, the result of \cite[Thm.1]{Bresler2008} for the complex Gaussian IC, which shows that the capacity is within a constant-gap of the deterministic IC capacity, depends on the uniformity of the optimal input distribution in the model to show that $I(W;Y_{2,\text{G}}^n)\leq I(W;Y_{2,\text{LDM}}^n)+cn$, where G stands for Gaussian model. However, it was shown in \cite{Mukherjee2017} that an integer-input integer-output model of the MAC-WT and WT-H, is within a constant-gap of the G-MAC-WT and G-WT-H. We therefore utilize a variation of that model in section~\ref{Developing a Converse from LD-Bounds} to transfer the results from Theorem~\ref{THM-converse-LD-WT-Helper} to the Gaussian case, which results in the following theorem:
\begin{theorem}\label{THM-converse-G-WT-Helper}
The secrecy rate $R$ of the Gaussian wiretap channel with one helper and symmetric channel gains at the wiretapper is bounded from above by
\begin{equation*}
R \leq \min\{ r_{ub1}, r_{ub2}, r_{ub3}\}+c
\end{equation*}
with \begin{IEEEeqnarray*}{rCl}
r_{ub1}&=& n_p+\tfrac{1}{2}n_c+\tfrac{1}{2}(n_1-n_2)^+\\
r_{ub2}&=&n_{1}\\
r_{ub3}&=&n_{2}+(n_{1}-n_{2}-n_{E})^+\\
&& +\:[n_E-n_{2}-(n_E-n_{1}+n_{2})^+]^+,
\end{IEEEeqnarray*}
where c is a constant independent of the signal-to-noise ratio.
\end{theorem}
We want to point out, that the bound $r_{ub3}$ has two possible forms for $n_E>n_2$ (vanishing private part), depending on the relation between $n_1$ and $n_2$. For $n_1\leq 2n_2$, $r_{ub3}=n_2$, while for $n_1>2n_2$ the bound becomes $r_{ub3}=n_1-n_2$, which can be seen in Fig.~\ref{beta_1}.

\subsection{Results for the G-MAC-WT}
For the Gaussian multiple access wiretap channel, we use the same techniques as in the G-WT-H case. This means we utilize lattice codes to transfer the achievable scheme from the linear deterministic model to the Gaussian model, which results in the following theorem:
\begin{theorem}
An achievable secrecy sum-rate $R_{\Sigma}$ of the Gaussian multiple-access wiretap channel is
\begin{equation}
R_{\Sigma}=r^p+r^c+r^R  
\end{equation}
where $r^c\!:=\!l_{\text{u}}(\tfrac{1}{2}\log \text{SNR}_1^{(1-\beta_1)}\!-\!\tfrac{1}{2})$, with\\
\begin{equation}
l_{\text{u}}:= 2\left \lfloor \frac{\min\{1+\beta_2-\beta_1,1\}}{3(1-\beta_1)}\right\rfloor,
\end{equation} 
$r^p:=\tfrac{1}{2}\log(\max\{1,\text{SNR}_1^{\beta_1-\beta_2}\})$, and
\begin{equation}
r^R=\begin{cases}
r^{R_1} & \text{for } r^{R_1}< r^{R_2}
\\
r^{R_2} & \text{for } 2r^{R_2} > r^{R_1}\geq  r^{R_2}\\
r^{R_1}+r^{R_2} & \text{for } r^{R_1} \geq 2r^{R_2}.
\end{cases}
\end{equation}
with \begin{IEEEeqnarray*}{rCl}
r^{R_1}&:=&\tfrac{1}{2}\log \text{SNR}_1^{1-\tfrac{3}{2}l_\text{u}(1-\beta_1)} - \tfrac{1}{2}\log \text{SNR}_1^{\min\{ \beta_1-\beta_2,0\}}-\tfrac{1}{2}\\
r^{R_2}&:=&\tfrac{1}{2}\log \text{SNR}_1^{(1-\beta_1)}-\tfrac{1}{2}.
\end{IEEEeqnarray*}
\end{theorem}
As for the G-WTH-H case, we can see that there is a direct correspondence between the achievable rate for the linear deterministic model and the results for the Gaussian model. In particular, the rate is within a constant bit-gap of the rates of the deterministic model, by comparing via $n=\lfloor \tfrac{1}{2}\log \text{SNR}\rfloor$.
Moreover, we transferred the upper bound of the LD-MAC-WT to the G-MAC-WT, which results in the following theorem.

\begin{theorem}\label{THM-converse-LD-MAC-WT}
The secrecy sum-rate $R_{\Sigma}$ of the Gaussian multiple access wiretap channel with symmetric channel gains at the eavesdropper is bounded from above by
\begin{equation}
R_{\Sigma}\leq\begin{cases}
n_p+\tfrac{2}{3}n_c+\frac{1}{3}n_\Delta+c & \text{for } n_2\geq n_E
\\
\tfrac{2}{3}n_c+\frac{1}{3}n_\Delta+c  & \text{for } n_E> n_2,
\end{cases}
\end{equation} where c is a constant independent of the signal-to-noise ratio.
\end{theorem}
\begin{figure}
\centering
\includegraphics[scale=0.63]{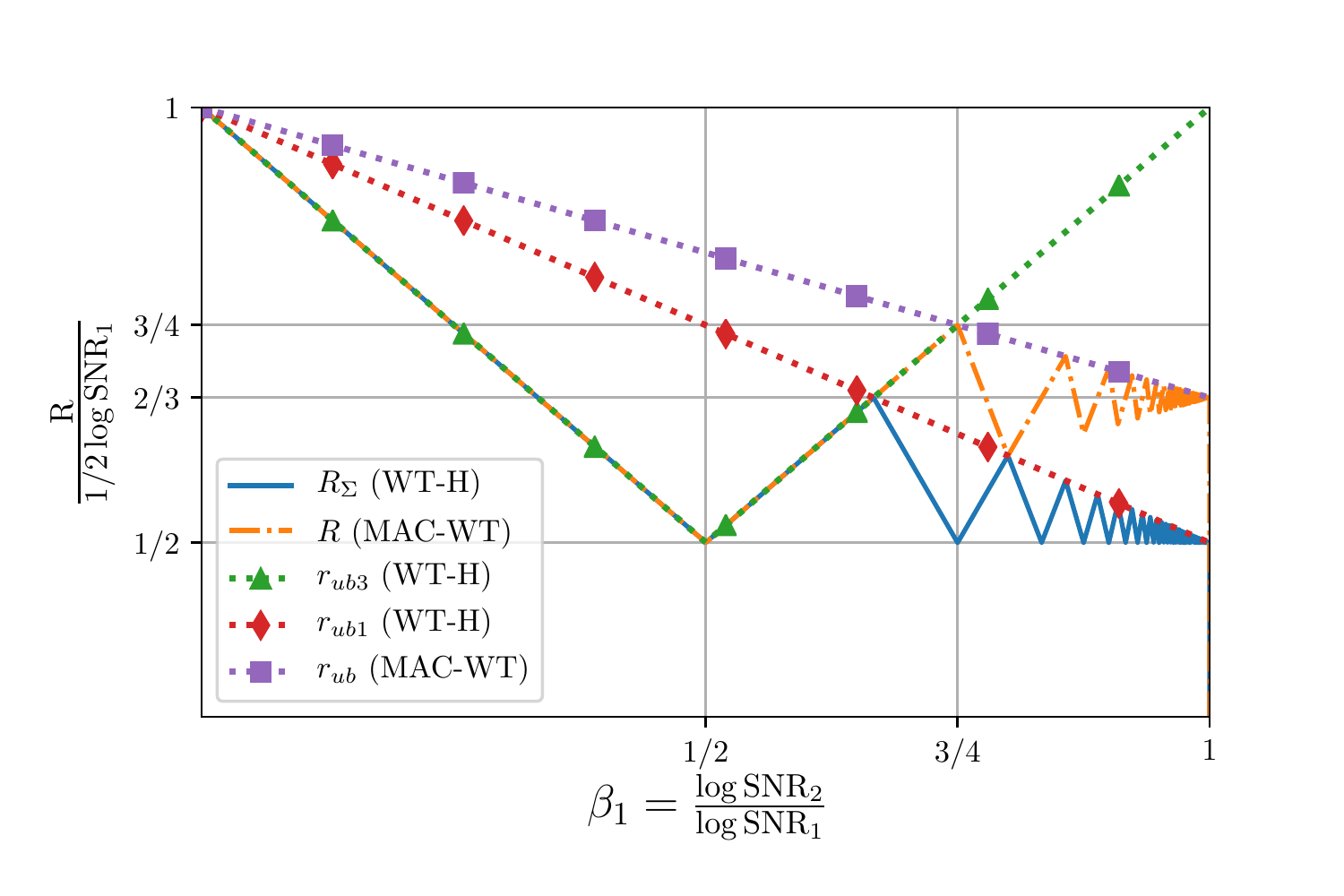}
\caption{Illustration of the achievable secrecy rates and upper bounds for the Gaussian WT-H and the  Gaussian MAC-WT in relation to the single-link scenario, i.e. normalized by $\log \text{SNR}_1$, and variation in the $\beta_1$ parameter, while $\beta_2$ is fixed at 1, i.e. vanishing private part.}
\label{beta_1}
\end{figure}
\subsection{Discussion}
Looking into Fig.~\ref{beta_1}, one can see the achievable rate normalized by the single-link signal-to-noise ratio, with varying parameter $\beta_1$ between $0$ and $1$, i.e. channel gain configurations. We note, that the theorems also give results for $\beta_1>1$, which are not shown since they would need separate figures due to scaling differences, but would basically show a mirrored (at the right hand side) picture. Moreover, it shows the results for $\beta_2=1$, which lets the private part vanish. Note that the private part just provides an off-set to the results above, and we therefore leave it out of the figure. Fig.~\ref{beta_1} therefore provides a look at the  secure g.d.o.f. of both models. One can see that the secure g.d.o.f converge to the s.d.o.f of $\frac{1}{2}$ for the G-WT-H, and $\frac{2}{3}$ for the G-MAC-WT for $\beta_1\rightarrow1$ and $\beta_2=1$, which agrees with the results of \cite{XieUlukusWiretap-Helper}. We can also see, that the achievable rate of both models fluctuates between the upper bound and a lower bound, for the part where the bit-level alignment scheme is dominant, which is in the range $\tfrac{2}{3}\leq \beta_1<1$ and $\tfrac{3}{4}\leq \beta_1<1$ for the WT-H and the MAC-WT, respectively. We believe that this is a result of the orthogonal bit-level alignment techniques which get transferred to the Gaussian model. A deterministic model with inter-dependent bit-levels, like the one used in \cite{Niesen-Ali}, could help to completely reach the upper bound, therefore resulting in a constant-gap sum-capacity result for that $\beta_1$ range. Moreover, as can be seen from the figure, we could not show a satisfying converse for the MAC-WT in the range of $0<\beta_1\leq \tfrac{3}{4}$. Comparing with the WT-H model, one would assume that a similar converse can be shown there. However, due to the jamming ability of both transmitters, which results in stochastic encoding functions, those converse techniques cannot be applied. Finally, we note that our scheme encounters a singularity at $\beta_1=1$, where no signal-scale diversity can be used and the alignment scheme fails. A possible solution could be to also utilize the channel phase, similar to \cite{CadambeJafarWangonMadsen}.

\section{Analysis of the Approximative Models}
\subsection{Proof of Theorem 1: Achievable Scheme for the Wiretap Channel with a Helper}
\label{ach-scheme-det-model}

{\bf Case $n_1\geq n_2$ :}
We denote the part of $\mathbf{X}_1$ and $\mathbf{X}_2$ in $\mathbf{Y}_{1,c}$ by $\mathbf{X}_{1,c}$ and $\mathbf{X}_{2,c}$, respectively. Moreover, we partition these common parts of the signals into $2n_\Delta$-bits partitions. We now utilize the first $n_\Delta$ bits of every {\it full} partition in $\mathbf{X}_{1,c}$ for messages and leave the remainder free. And for $\mathbf{X}_{2,c}$ we utilize the first $n_\Delta$-bits of every partition for jamming, while the rest is free. After partitioning, $\mathbf{Y}_{1,c}$ has a remainder part with $n_Q$ bit-levels.
The user signal in this remainder part follows the same rules as before, while the helper lets the first $n_\Delta$ bits free and only utilizes the bits afterwards for jamming, until we have filled all $n_Q$ bits. The private part $\mathbf{Y}_{1,p}$ can be used completely by the user, and all of $\mathbf{X}_1$ in this part can be used for messaging. The total achievable rate is the private rate $r_p=n_p$ plus the common rate
\begin{equation}
r_c=\frac{1}{2}\left(\left\lfloor \frac{n_c}{2n_\Delta}\right\rfloor2n_\Delta\right)+Q,
\end{equation}
where $Q$ is defined as in the theorem. The common rate follows from the fact that we utilize half of the bits of all $2n_\Delta$ partitions, along with a remainder part $Q$. In the remainder we utilize every bit, as long as $n_Q$ is smaller than $n_\Delta$. If $n_Q$ is larger than $n_\Delta$, we only utilize the first $n_\Delta$ bits.

{\bf Case $n_1 <n_2$: }
We use the same strategy as before, except for the remainder part $n_Q$ of $\mathbf{Y}_{1,c}$. In the remainder part, the first $n_\Delta$ bits of the user are left free, and all bits afterwards are used for messaging, while the Helper only jams the first $n_\Delta$ bits. The strategy is therefore the opposite as before. This yields a different $Q$-term, where for $n_Q<n_\Delta$ no rate is achieved, and for $n_Q\geq n_\Delta$ one can use $n_Q-n_\Delta$ bits for messaging. We note that the secrecy is provided by the (Crypto-) Lemma \ref{Crypto_lemma} and the fact that we use binary addition on each level as well as jamming signals chosen such that each bit is Bern($\tfrac{1}{2}$) distributed. And we therefore have that $I(\mathbf{X}^n_1;\mathbf{Z}^n)=0$.
\qed
\begin{lemma}[Crypto-Lemma, \cite{forney2004}]
Let $G$ be a compact abelian group with group operation $+$, and
let $Y = X + N$, where $X$ and $N$ are random variables over $G$ and $N$ is independent of $X$ and
uniform over $G$. Then $Y$ is independent of $X$ and uniform over $G$.
\label{Crypto_lemma}
\end{lemma}

\begin{remark}The bit-level shift between $\mathbf{X}_1$ and $\mathbf{X}_2$ of $n_\Delta$ bits makes it impossible to divide $\mathbf{Y}_1$ in exclusively private and common parts. In our division, the bottom $n_\Delta$ bits of $\mathbf{x}_{1,c}$ are only received at $\mathbf{Y}_1$ and are, therefore, private. Hence, the common rate $r_c$ is not purely made of common signal parts. Nevertheless, our choice of division reaches the upper bound and fits into the scheme.
\label{remark1}
\end{remark}

 \subsection{Proof of Theorem 3: Achievable Scheme for the LD-MAC-WT}\label{Ach_Scheme_LD_MAC_WT}
 For the LD-MAC-WT, due to symmetry, we may assume w.l.o.g. that $n_1>n_2$, where we leave out the case that $n_1=n_2$, see remark \ref{remark2}.
First of all, we look at the case that $n_2\geq n_E$. Our strategy is the same as before, i.e. to deploy a cooperative jamming scheme such that minimal jamming is done to $\mathbf{Y}_{1,c}$, while maximal jamming is received at $\mathbf{Y}_2$. We partition the common signals, $\mathbf{X}_{1,c}$ and $\mathbf{X}_{2,c}$, into $3n_\Delta$-bit parts and partition these parts again into $n_\Delta$-bit parts. For $\mathbf{X}_{1,c}$, in every $3n_\Delta$-bit part we use the first $n_\Delta$ bits for the message and the next $n_\Delta$ bits for jamming, while the last $n_\Delta$ bits will not be used. For $\mathbf{X}_{2,c}$, in every $3n_\Delta$-bit part, the first $n_\Delta$ bits will be used for jamming. The next $n_\Delta$ bits will be used for the message and the last $n_\Delta$ bits are left free. There will be a remainder part with 
\begin{equation}
n_Q= n_c \bmod 3n_\Delta \text{ bits}.
\end{equation}
The remainder part follows the same design rules as the $3n_\Delta$ parts, except that $\mathbf{X}_{2,c}$ leaves the first $n_\Delta$ bits free, then uses jamming on the next $n_\Delta$ bits and utilizes the last $n_\Delta$ bits for messaging, until $n_Q$ bits are allocated. The scheme is designed such that the jamming parts of $\mathbf{X}_{1,c}$ and $\mathbf{X}_{2,c}$ overlap at $\mathbf{Y}_{1,c}$, while the message parts of one signal overlap with the non-used part of the other signal. However, due to the signal strength difference $n_\Delta$, the jamming parts overlap with the messages at $\mathbf{Y}_{2}$, see Fig.~\ref{Design_Scheme_2}. Secure communication is therefore provided by the Crypto-lemma, as long as we use a Bern$(\tfrac{1}{2})$ distribution for the jamming bits. The whole private part can be used for messaging and its sum-rate is therefore $r_p=n_p$.
The achievable secure rate for the common part consists of the rate for the $3n_\Delta$ partitions and the remainder part. It can be seen that every $3n_\Delta$-part of $\mathbf{Y}_{1,c}$ allocates $2n_\Delta$ bits for the messages. This results in the common secrecy rate
\begin{equation}
r_c= (\lfloor \tfrac{n_c}{3n_\Delta} \rfloor3n_\Delta)\tfrac{2}{3} +Q,
\end{equation}
where $Q$ specifies the rate part of the remainder term. In the remainder part we allocate all remaining bits as message bits, as long as $n_Q < n_\Delta$. For $2n_\Delta > n_Q \geq  n_\Delta$, we allocate the first $n_\Delta$ bits of $n_Q$ for the message. And for $n_Q \geq 2n_\Delta$, we allocate the first $n_\Delta$ bits as well as the last $q$ bits, where $q$ is defined as
\begin{equation}
q=n_Q \bmod n_\Delta.
\end{equation}
This results in
\begin{equation}
Q=\begin{cases}
q & \text{for } n_Q< n_\Delta 
\\
n_\Delta & \text{for } 2n_\Delta > n_Q \geq  n_\Delta\\
n_\Delta + q & \text{for } n_Q \geq 2n_\Delta.
\end{cases}
\end{equation}
Together with the private rate term, we achieve
\begin{equation*}
R=\tfrac{2}{3}(\lfloor \tfrac{n_c}{3n_\Delta} \rfloor3n_\Delta)+n_p+Q.
\end{equation*}
For $n_2\geq n_E$ the achievable scheme is the same, except that we do not have a private part. We therefore have an achievable rate of 
\begin{equation*}
R=\tfrac{2}{3}(\lfloor \tfrac{n_c}{3n_\Delta} \rfloor3n_\Delta)+Q,
\end{equation*}
which completes the proof. \qed

\begin{figure}
\centering
\includegraphics[scale=0.75]{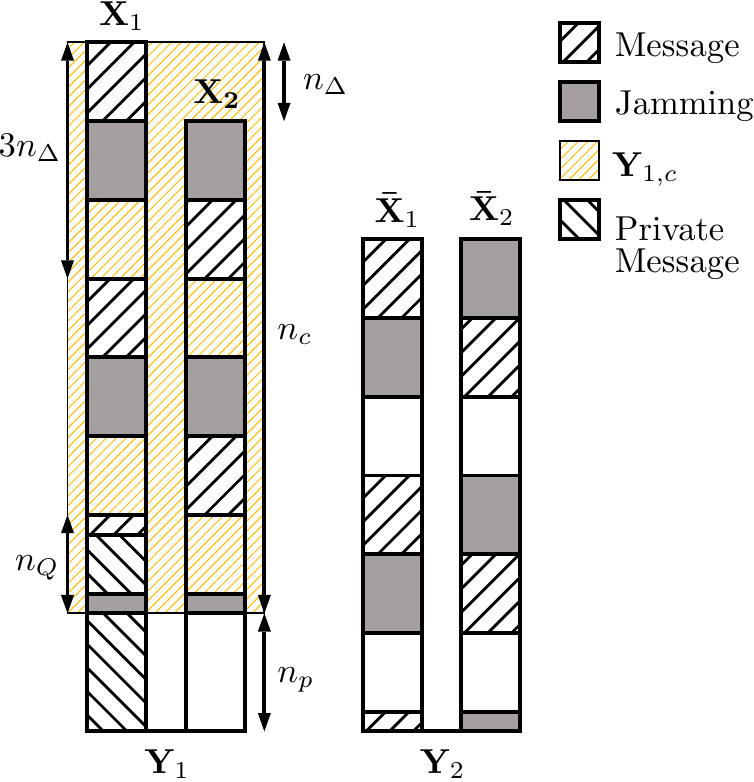}
\caption{Illustration of the achievable scheme for the LD-MAC-WT. The private part $\mathbf{Y}_{1,p}$ can be used freely and is allocated by User 1. The common part $\mathbf{Y}_{1,c}$ uses our alignment strategy. The strategy exploits the channel gain difference between both signals, to minimize the effect of jamming at the receiver $\mathbf{Y}_1$, while jamming all signal parts at the eavesdropper $\mathbf{Y}_2$.}
\label{Design_Scheme_2}
\end{figure}

\subsection{Proof of Theorem 2 and 4: Upper Bounds for the LD-WT-H and the LD-MAC-WT}
\label{LD-Converse}
The converse proofs for the Gaussian models are developed such that we can use the upper bounds of the rates for the linear deterministic models. We therefore leave out those proofs for the deterministic models due to space limitations. One can reconstruct the proofs by substituting the terms of the truncated deterministic model with the linear deterministic ones. The following table gives the correspondence between the terms: \begin{center}
\bgroup
\def\arraystretch{1.3}%
  \begin{tabular}{ c | c }
  	Truncated Model & LD Model\\
    \hline
    $Y_{1,\text{D}},Y_{2,\text{D}}$ &  $\mathbf{Y}_1,\mathbf{Y}_2$ \\ \hline
    $\lfloor h_{1}X_{1,\text{D}}\rfloor$ & $\mathbf{X}_1$ \\ \hline
    $\lfloor h_{2}X_{2,\text{D}} \rfloor$ &  $\mathbf{X}_2$ \\ \hline
    $\lfloor h_{E}X_{2,\text{D}} \rfloor$ & $\mathbf{\bar{X}}_2$\\ \hline
    $\lfloor h_{E}X_{1,\text{D}}\rfloor$ & $\mathbf{\bar{X}}_1$\\ \hline
  \end{tabular}
\egroup
\end{center}
We note that the indexing via levels, i.e. such as $(\mathbf{\bar{X}}_1)_{[:n_2]}$ is the same for both models.

\begin{remark} Our schemes rely on the signal strength difference between both users. Our scheme would not work, if $n_1=n_2$, while having equal channel gains at the eavesdropper. In that case we would not have any signal strength diversity to exploit which results in a singularity point where the secrecy rate is zero.
\label{remark2}
\end{remark}

\section{Analysis of the Gaussian wiretap channel with a helper}
\subsection{Proof of Theorem 5: Achievable Rate for the G-WT-H}
\label{WT_Helper_Achievable_Scheme}
The following procedure is inspired by the achievability proof in \cite{Bresler2010}.
We will look into the case that $\text{SNR}_1\geq \text{SNR}_2$, the other case follows similarly. For the achievable scheme, we need to partition the received signal-to-noise ratios at $Y_1$ into intervals of $\text{SNR}_1^{(1-\beta_1)}$. Each of these intervals plays the role of an $n_\Delta-$Interval of bit-levels in the linear deterministic scheme and has a power-to-noise ratio $\theta_l$, which is defined as 
\begin{IEEEeqnarray}{rCl}
\theta_{l} & = & q_{l-1}-q_{l}=\text{SNR}_1^{1-(l-1)(1-\beta_1)}-\text{SNR}_1^{1-l(1-\beta_1)}\IEEEyesnumber
\label{theta_k_fractions}
\end{IEEEeqnarray}
with $l$ indicating the specific level. The users decompose the signals $X_i$ into a sum of independent sub-signals $X_i=\sum_{l=1}^{l_{max}}X_{il}$. We will use $n$-dimensional nested lattice codes, introduced in \cite{UrezZamir}, which can achieve capacity in the AWGN single-user channel. 
\subsubsection{Nested Lattice Codes}
A lattice $\Lambda$ is a discrete subgroup of $\mathbb{R}^n$ which is closed under real addition and reflection. Moreover, denote the nearest neighbour quantizer by 
$
Q_{\Lambda}(\mathbf{x}):=\arg\min_{\mathbf{t}\in \Lambda} ||\mathbf{x}-\mathbf{t}||.
$ The fundamental Voronoi region $\mathcal{V}(\Lambda)$ of a lattice $\Lambda$ consists of all points which get mapped or quantized to the zero vector and the modulo operation for lattices is defined as
$
[\mathbf{x}]\bmod\Lambda := \mathbf{x}-Q_{\Lambda}(\mathbf{x}).
$

A nested lattice code is composed of a pair of lattices $(\Lambda_{\text{fine}},\Lambda_{\text{coarse}})$, where $\mathcal{V} (\Lambda_{\text{coarse}})$ is the fundamental Voronoi region of the coarse lattice and operates as a shaping region for the corresponding fine lattice $\Lambda_{\text{fine}}$. It is therefore required that $\Lambda_{\text{coarse}} \subset \Lambda_{\text{fine}}$. Such a code has a corresponding rate $R$ equal to the logarithm of the nesting ratio. A part of the split message is now mapped to the corresponding codeword $\mathbf{u}_i(l)\in \Lambda_{\text{fine},l-1} \cap \mathcal{V} (\Lambda_{\text{coarse},l})$, which is a point of the fine lattice inside the fundamental Voronoi region of the coarse lattice. Note that $\Lambda_{l_{max}} \subset \cdots\subset \Lambda_1$. 
The code is chosen such that it has a power of $\theta_{l}$. The codeword $\mathbf{x}_{i}(l)$ is now given as 
$
\mathbf{x}_{i}(l)=[\mathbf{u}_{i}-\mathbf{d}_{i}] \bmod  \Lambda_l,
$
where we dither (shift) with $\mathbf{d}_{i}\sim \mbox{Unif} (\mathcal{V}(\Lambda_l))$ and reduce the result modulo-$\Lambda_l$. Transmitter $i$ now sends a scaled $\mathbf{x}_{i}$ over the channel, such that the power per sub-signal $\mathbf{x}_{i}(l)$ is $\frac{\theta_{l}}{|h_{i1}|^2}$ and receivers see a power of $\theta_{l}$. Due to the partitioning construction, the $\mathbf{x}_{i}$ satisfy the power restriction of $P$ for user 1,
\begin{equation}
\sum\limits_{l=1}^{l_{\max}} \frac{\theta_{l}}{|h_{11}|^2} \leq \frac{\mbox{SNR}_{1}}{|h_{11}|^2} = P
\end{equation}
and user 2
\begin{equation}
\sum\limits_{l=2}^{l_{\max}} \frac{\theta_{l}}{|h_{21}|^2} \leq \frac{\mbox{SNR}_{2}}{|h_{21}|^2} = P.
\end{equation}

Moreover, aligning sub-signals use the same code (with independent shifts). In \cite{UrezZamir} it was shown that nested lattice codes can achieve the capacity of the AWGN single-user channel with vanishing error probability. Viewing each of our power intervals as a channel, we therefore have that
\begin{equation}
R(l)\leq \tfrac{1}{2}\log \left(1+\frac{\theta_{l}}{N(l)}\right),
\label{Decodingbound}
\end{equation}
where $N(l)$ denotes the noise variance per dimension of the sub-sequent levels. Now, $X_1$ is used for signal transmission, while $X_2$ is solely used for jamming. As in the deterministic case, the objective is to align the signal parts of $X_1$ with the jamming of $X_2$ at $Y_2$, while allowing decoding of the signal parts at $Y_1$. Due to the signal scale based coding strategy and the equal received signal-to-noise ratios at $Y_2$, an alignment is achieved with the proposed scheme. We use a jamming strategy, where the jamming sub-codeword is uniformly distributed on $\mathcal{V}$, therefore $\mathbf{x}_{2,\text{jam}}(l)\sim \mbox{Unif} (\mathcal{V}(\Lambda_l))$. Now, application of lemma \ref{Crypto_lemma} shows, that the received codeword $\mathbf{y}= [\mathbf{x}_1(l)+\mathbf{x}_{2,\text{jam}}(l)] \bmod \Lambda_l$ is independent of $\mathbf{x}_1(l)$, therefore providing secrecy. The only thing left to prove is that the signal can be decoded at $Y_1$.
The decoding is done level-wise, treating subsequent levels as noise. Every level is treated as a Gaussian point-to-point channel with power $\theta_l$ and noise $1+\text{SNR}_1^{1-l(1-\beta_1)}$, which consists of the base noise $N_1$ at $Y_1$ and the power of all subsequent levels of both signals. Successful decoding can be assured with a rate limitation (see \eqref{Decodingbound}) of 
\begin{equation}
r_l\leq \tfrac{1}{2}\log \left(1+\frac{\theta_l}{1+\text{SNR}_1^{1-l(1-\beta_1)}}\right).
\label{Decodingbound_actual}
\end{equation}
\subsubsection{Achievable rate}
As in the deterministic case, we have a private and common part and are defined equivalently. The common part depends on the strength of the received power at the eavesdropper ($n_c:=\min\{n_E+n_\Delta,n_1\}$ for $n_1\geq n_2$). Note that there is the correspondence \begin{equation}\label{beta_1_corr}
\beta_1=\frac{\log \text{SNR}_2}{\log \text{SNR}_1}\approx \frac{n_2}{n_1}
\end{equation} and also $\beta_2 \approx \tfrac{n_E}{n_1}$. Therefore, the part $n_E+n_\Delta$ corresponds to $\text{SNR}_1^{\beta_2+(1-\beta_1)}$ in the Gaussian model. The opposing remainder is therefore $\text{SNR}_1^{\beta_1-\beta_2}$ and we get the common power-to-noise ratio as 
\begin{equation}
\text{SNR}_c:=\text{SNR}_1-\max\{1,\text{SNR}_1^{\beta_1-\beta_2}\},
\end{equation}
 while the private part has a power-to-noise ratio of 
 \begin{equation}
\text{SNR}_p:=\max\{1,\text{SNR}_1^{\beta_1-\beta_2}\}-1.
 \end{equation} The private part will not be partitioned further, since it can be used completely and without penalty. Moreover, it has only the base noise and a rate of $r^p=\tfrac{1}{2}\log(1+\text{SNR}_p)$ can be achieved. For the common part, we use the deterministic achievable scheme and need to partition the available power-to-noise ratio. 
All odd levels $l$ of $X_1$  will be used for signal transmission. Every level $l$ can handle a rate of $r_l$. We can simplify the rate of \eqref{Decodingbound_actual} with
\begin{IEEEeqnarray*}{rCl}
\log \left(1+\frac{a-b}{1+b} \right)&= & \log \left(\frac{a+1}{1+b} \right)
\geq \log \left(\frac{a}{2b} \right)
\end{IEEEeqnarray*} where we used that $b>1$ to get $r_l\geq \tfrac{1}{2}\log \text{SNR}_1^{(1-\beta_1)}-\tfrac{1}{2}$, where the $1-l(1-\beta_1)$-terms get cancelled in the last fraction. Since we use the same scheme as in the deterministic case, we have a total of
\begin{equation}
l_{\text{u}}:= \left\lfloor \frac{\min\{1+\beta_2-\beta_1,1\}}{2(1-\beta_1)}\right\rfloor
\end{equation}
used levels in $X_1$, where the remainder term is not yet included. The alignment section of the common part has a total rate of
\begin{equation}
r^c=l_{\text{u}}(\tfrac{1}{2}\log \text{SNR}_1^{(1-\beta_1)}-\tfrac{1}{2}),
\label{Gauss_rate}
\end{equation}
which corresponds to $\lfloor \tfrac{n_c}{2n_\Delta}\rfloor n_\Delta $ in the deterministic case. Moreover, we need to consider the remainder term, which is allocated between the alignment structure and the noise floor or the private part. Once again we use the deterministic scheme as a basis. We see in \eqref{reminder_ld_helper}, that we have two cases for $n_1\geq n_2$, which corresponds to $\text{SNR}_1\geq \text{SNR}_2$. If the remainder has a power-to-noise ratio of 
\begin{IEEEeqnarray*}{rCl}
\text{SNR}_r&=&\text{SNR}_1^{1-2l_\text{u}(1-\beta_1)}-\text{SNR}_1^{\min\{ \beta_1-\beta_2,0\}}\\
&&<\:\text{SNR}_1^{1-2l_\text{u}(1-\beta_1)}-\text{SNR}_1^{1-(2l_\text{u}+1)(1-\beta_1)},
\end{IEEEeqnarray*} it achieves a rate of
\begin{equation*}
r^R \geq \tfrac{1}{2}\log \text{SNR}_1^{1-2l_\text{u}(1-\beta_1)} - \tfrac{1}{2}\log \text{SNR}_1^{\min\{ \beta_1-\beta_2,0\}}-\tfrac{1}{2}.
\end{equation*} Otherwise, it can only use a full partition, which leads to 
\begin{equation*}
r^R \geq \tfrac{1}{2}\log \text{SNR}_1^{(1-\beta_1)}-\tfrac{1}{2}.
\end{equation*}

We therefore get a total rate of 
$
r_{\text{ach}}=r^p+r^c+r^R.
$
The case for $\text{SNR}_2> \text{SNR}_1$ can be shown similarly.
 \qed

\subsection{Proof of Theorem 6: Developing a Converse from LD-Bounds}\label{Developing a Converse from LD-Bounds} 
Ideally, we want to re-use the insights (and results) from the deterministic model to derive an upper bound for the Gaussian channel. However, there is no direct constant-gap relation for the relevant models between the Gaussian channel version and the linear deterministic model. It was shown in \cite{Mukherjee2017} that the integer-input integer-output model of the MAC-WT, is within a constant-gap of the G-MAC-WT. This result can also be used for the G-WT-H. The system equations for the integer-input integer-output model can be written as 
\begin{IEEEeqnarray}{rCl}\label{Integer_Model}
\bar{Y}_{1,\text{D}}&=&\lfloor h_{11}\bar{X}_{1,\text{D}}\rfloor +\lfloor h_{21}\bar{X}_{2,\text{D}} \rfloor\IEEEyessubnumber\\
\bar{Y}_{2,\text{D}}&=&\lfloor h_{22}\bar{X}_{2,\text{D}} \rfloor + \lfloor h_{12}\bar{X}_{1,\text{D}}\rfloor \IEEEyessubnumber,
\end{IEEEeqnarray} where the $\bar{X}_1^D\in \{0,1,\ldots,\lfloor \sqrt{\text{SNR}}\rfloor\}$. One can construct these codewords easily from a set of given codewords for the Gaussian case by $\lfloor X_{\text{G}}\rfloor \bmod \lfloor \sqrt{P}\rfloor$. It was shown in \cite{Mukherjee2017}, that the mutual information terms for the integer-input integer-output channel \eqref{Integer_Model} are within a constant-gap\footnote{It was actually stated that both terms are within $o(\log P)$. However, the result also satisfies the stronger notion of a constant-gap.} of the corresponding Gaussian model \eqref{Gauss_Model_Helper}, which means that
\begin{IEEEeqnarray}{rCl}\label{Constant_Gap}
I(W_1,W_2;\bar{Y}_{1,\text{D}}^n)&\leq& I(W_1,W_2;Y_{1,\text{G}}^n)+nc\IEEEyessubnumber\\
I(W_1,W_2;Y_{2,\text{G}}^n)&\leq& I(W_1,W_2;\bar{Y}_{2,\text{D}}^n)+nc\IEEEyessubnumber,
\end{IEEEeqnarray} where $c$ is a constant. The first inequality follows from a proof in \cite{Bresler2008} and a more detailed version of the same ideas in \cite{Davoodi2016}. The second inequality builds on lemmata and ideas from \cite{Bresler2008,Davoodi2016} and  \cite{Avestimehr2011}.
We therefore have that 
\begin{IEEEeqnarray*}{rCl}\label{Converse-Start}
n(R_1+R_2)&=& I(W_1,W_2;Y_{1,\text{G}}^n)-I(W_1,W_2;Y_{2,\text{G}}^n)+n\epsilon\\
&\leq & I(W_1,W_2;\bar{Y}_{1,\text{D}}^n)-I(W_1,W_2;\bar{Y}_{2,\text{D}}^n)\IEEEyesnumber\\
&&+n(c+\epsilon),
\end{IEEEeqnarray*} which shows that any bound for the integer-input integer-output model can be used as an outer bound for the corresponding Gaussian model. Now, to bring the LDM ideas to the truncated model, we modify the form such that $\bar{X}_1^D$ is represented\footnote{For the time being, we use $n$ as the index of the bit-level as well as the sequence index. This will be distinguishable later on, since the bit-level index will always have a subscript indicating the specific channel gain.} as
\begin{equation}
X_{1,D}=2^n\sum_{b=1}^n \tilde{X}_{1,b}2^{-b}\in \{0,1,\ldots,2^n-1\},
\end{equation}
where $n=\lfloor \log \lfloor\sqrt{\text{SNR}} \rfloor \rfloor$ and $\tilde{X}_{1,b}\in\mathbb{F}_2$. Note that the floor function around the logarithm, i.e. the quantization from integers to powers of two, reduces the cardinality of the input constellation by at most half plus one-half, which results in a maximum bit-gap of 2 bits in the capacity results for high-SNR. We can therefore work with the model 
\begin{IEEEeqnarray}{rCl}\label{Truncated-Model1}
Y_{1,\text{D}}&=&\lfloor h_{11}X_{1,\text{D}}\rfloor +\lfloor h_{21}X_{2,\text{D}} \rfloor\IEEEyessubnumber\\
Y_{2,\text{D}}&=&\lfloor h_{22}X_{2,\text{D}} \rfloor + \lfloor h_{12}X_{1,\text{D}}\rfloor \IEEEyessubnumber,
\end{IEEEeqnarray} where $h_{ij}X_{i,D}=h_{ij}2^{n_{ij}}\sum_{b=1}^{n_{ij}} \tilde{X}_{i,b}2^{-b}$, $h_{ij}\in [1,2)$ and the $n_{ij}$ correspond to the bit-levels in the LD model. We now include the simplifying notation changes from the deterministic model and the assumption on equal received power at the wiretaper and write the model as
\begin{IEEEeqnarray}{rCl}\label{Truncated_Model2}
Y_{1,\text{D}}&=&\lfloor h_{1}X_{1,\text{D}}\rfloor +\lfloor h_{2}X_{2,\text{D}} \rfloor\IEEEyessubnumber\\
Y_{2,\text{D}}&=&\lfloor h_{E}X_{2,\text{D}} \rfloor + \lfloor h_{E}X_{1,\text{D}}\rfloor \IEEEyessubnumber.
\end{IEEEeqnarray} 
 We will call this model the truncated deterministic model (TDM).
 For the converse proofs we will also need the following lemmata. Note that the following lemmata results and ideas were already used for example in the converse proof in \cite{CGeng2015} but without rigorous justification. Moreover, the first lemma uses ideas from a proof in \cite{Bresler2008}.

\begin{lemma}\label{Lemma_entropy_floorfunc_equal_bits}
For an arbitrary signal $X_{D}\in\{0,1,\ldots,2^n-1\}$, with $n\in \mathbb{N}$ and channel gain $h\in [1,2)$ we have that 
\begin{equation*}
H(\lfloor hX_D \rfloor)=H(\tilde{X}_1,\ldots,\tilde{X}_n),
\end{equation*} where $\tilde{X}_i\in \mathbb{F}_2$ are such that $X_{D}=2^n\sum_{i=1}^n \tilde{X}_i2^{-i}$.

\end{lemma}
\begin{proof}
We denote the tuple $(\tilde{X}_1,\ldots,\tilde{X}_n)\in \mathbb{F}_2^n$ by $\tilde{\mathbf{X}}$. There is a bijection $f_1:\mathbb{F}_2^n \rightarrow \{0,1,\ldots,2^n-1\}$ which can be constructed as $f_1(\tilde{\mathbf{X}})=2^n\sum_{i=1}^n \tilde{X}_i2^{-i}$. Now, the resulting integers are distance one apart. Therefore multiplying by $h\in [1,2)$ does not lower the distance. Quantizing those scaled values to the integer part only introduces gaps in the support, but does not reduce the cardinality. We therefore have that $f_2(X_D)=\lfloor h X_D\rfloor$ is again a bijection. Therefore, the composition of both functions $f_3=f_2 \circ f_1$ is a bijection and we have that 
\begin{equation*}
H(f_3(\tilde{\mathbf{X}}))=H(\tilde{\mathbf{X}})
\end{equation*} which shows the result.
\end{proof}

\begin{lemma}\label{Lemma_floorsplitting}
For an arbitrary signal $X_{D}\in\{0,1,\ldots,2^n-1\}$, with $n,m\in \mathbb{N}$, $m<n$, $X_{D}=2^n\sum_{i=1}^n \tilde{X}_i2^{-i}$, $\tilde{X}\in \mathbb{F}_2$ and channel gain $h\in [1,2)$ we have that 
\begin{IEEEeqnarray*}{rCl}
&&H(\lfloor h 2^n\sum_{i=1}^n \tilde{X}_i2^{-i} \rfloor)\\
&&=\:H(\lfloor h 2^n\sum_{i=1}^m \tilde{X}_i2^{-i} \rfloor+ \lfloor h 2^n\sum_{i=m+1}^n \tilde{X}_i2^{-i} \rfloor)
\end{IEEEeqnarray*}
\end{lemma}
\begin{proof}
The first entropy term contains $2^n\sum_{i=1}^n \tilde{X}_i2^{-i} \in \{0,1,\ldots,2^n\!-\!1\}$. As argued previously, the support has distance one, and multiplying by the channel gain and taking the integer part only introduces gaps in the support and scales the values up, but the cardinality stays the same. Therefore, $|\supp(X_D)|=|\supp(\lfloor hX_D\rfloor)|=2^n$. Now, the same is true for 
\begin{equation*}
\underline{X}_D:=2^n\sum_{i=m+1}^n \tilde{X}_i2^{-i}\in \{0,1,\ldots, 2^{n-m}\!-\!1\}.
\label{Lemma3_EQ1}
\end{equation*}
It also holds for 
\begin{IEEEeqnarray*}{rCl}
&&\bar{X}_D:=2^n\sum_{i=1}^m \tilde{X}_i2^{-i}\\
&&\in\: \{0,2^{n-m}, 2^{n-m+1},2^{n-m}+2^{n-m+1},\ldots,2^n\!-\!2^{n-m}\},
\label{Lemma3_EQ2}
\end{IEEEeqnarray*}
where the distance is $2^{n-m}> 1$, since $n>m$. Moreover, we have that $X_D=\bar{X}_D+\underline{X}_D$. The cardinality of the support of $\bar{X}_D$ is 
\begin{equation*}
\supp(\bar{X}_D)=\frac{2^n-2^{n-m}}{2^{n-m}}+1=2^m.
\end{equation*} 
Now, due to the structure\footnote{In particular because the biggest element of $\underline{X}_D$ is still smaller than the smallest distance in $\bar{X}_D$.}, the sum between $\underline{X}_D$ and $\bar{X}_D$ yields a Cartesian product between the support sets, and we therefore have that
\begin{IEEEeqnarray*}{rCl}
|\supp(\underline{X}_D+\bar{X}_D)|=2^n&=&2^{n-m}2^m\\
&=&|\supp(\underline{X}_D)||\supp(\bar{X}_D)|,
\end{IEEEeqnarray*}
for the support of the sum-set. The same holds for the scaled integer parts, since they have the same scaling and therefore
\begin{IEEEeqnarray*}{rCl}
|\supp(\lfloor h(\underline{X}_D+\bar{X}_D) \rfloor)|&=&|\supp(\underline{X}_D+\bar{X}_D)|\\
&=&|\supp(\underline{X}_D)||\supp(\bar{X}_D)|\\
&=&|\supp(\lfloor h\underline{X}_D\rfloor )||\supp(\lfloor h\bar{X}_D \rfloor)|\\
&=&|\supp(\lfloor h\underline{X}_D\rfloor +\lfloor h\bar{X}_D \rfloor)|,
\end{IEEEeqnarray*} which proves the result.
\end{proof} Moreover, we introduce the function 
\begin{IEEEeqnarray*}{rCl}
f_{[a:b]}(\lfloor hX_D \rfloor )&=&(\lfloor h X_D\rfloor)_{[a:b]}=\lfloor h_{ij}2^{n_{ij}}\sum_{k=a}^{b} \tilde{X}_k2^{-k}\rfloor,
\end{IEEEeqnarray*} which restricts the exponents of the binary expansion inside the term to lie in the set of integers $\{a,a+1,\ldots,b\}$.
The result of Lemma \ref{Lemma_floorsplitting} can then be written as \begin{equation*}
H(\lfloor h X_D\rfloor )=H((\lfloor h X_D \rfloor )_{[1:m]}+(\lfloor h X_D \rfloor)_{[m+1:n]}).
\end{equation*}
 If the term is a sum of two signals, then both get restricted relative to the stronger part. Therefore, a signal
\begin{equation}
Y_D=\lfloor h_1 2^n \sum_{i=1}^{n} \tilde{X}_{1,i}2^{-i} \rfloor + \lfloor h_2 2^m \sum_{i=1}^{m} \tilde{X}_{2,i}2^{-i} \rfloor,
\end{equation} where $n>m$, can be restricted to 
\begin{equation*}
(Y_D)_{[1:a]}=\lfloor h_1 2^n \sum_{i=1}^{a} \tilde{X}_{1,i}2^{-i} \rfloor + \lfloor h_2 2^m \sum_{i=1}^{a-(n-m)} \tilde{X}_{2,i}2^{-i} \rfloor.
\end{equation*} Moreover, we use the notation also on the bit-tuples to indicate that $(\tilde{X}_1,\ldots,\tilde{X}_n)\in \mathbb{F}_2^n$ by $\tilde{\mathbf{X}}$ is restricted to the bits $a$ to $b$, such that $(\tilde{X}_a,\ldots,\tilde{X}_b)$ is denoted as $(\tilde{\mathbf{X}})_{[a:b]}$. The notation is therefore the same as for the bit-vectors in the linear deterministic model.
We start with the G-WT-H version of equation \eqref{Converse-Start} and convert the steps of the proof for the linear deterministic case to the truncated deterministic model. 
\begin{IEEEeqnarray*}{rCl}
n(R-\epsilon)&=& I(W;Y_{1,\text{G}}^n)-I(W;Y_{2,\text{G}}^n)\\
&\leq & I(W;Y_{1,\text{D}}^n)-I(W;Y_{2,\text{D}}^n)+nc\\
&\leq & I(W;Y_{1,\text{D}}^n)-I(W;(Y_{2,\text{D}})_{[1:n_2]}^n)+nc\\
&= & H(Y_{1,D}^n)-H(Y_{1,D}^n|W)-H((Y_{2,D}^n)_{[:n_2]})\\
&&\: +H((Y_{2,D}^n)_{[:n_2]}|W)+nc\\
&=& H(Y_{1,D}^n)-H((Y_{2,D}^n)_{,[:n_2]})+H((\lfloor h_{E} X_{2,D}^n\rfloor )_{[:n_2]})\\
&&-\:H(\lfloor h_{2}X_{2,D}^n\rfloor )+nc,
\end{IEEEeqnarray*}
where Fanos inequality and the secrecy constraint was used. Moreover, we used the fact that $I(W;Y_{2,D}^n)\geq I(W;f(Y_{2,D}^n))$ for arbitrary functions $f$, due to the data processing inequality. Note that for $n_2\geq n_E$, we have that $(Y_{2,D}^n)_{[:n_2]}=Y_{2,D}^n$. In the last line we used that $X_{1,D}$ is a function of $W$, and $X_{2,D}$ is independent of $W$, due to the Helper model assumptions. We remark that the first property does not hold in general, since jamming through the first user would result in a stochastic function. Now, for $n_E\geq n_2$, both terms $\lfloor h_{2}X_{2,D}^n\rfloor$ and 
$(\lfloor h_{E} X_{2,D}^n\rfloor )_{[:n_2]}$ have the same bits, and we can use lemma \ref{Lemma_entropy_floorfunc_equal_bits} to show that
\begin{equation*}
H((\lfloor h_{E} X_{2,D}^n\rfloor )_{[:n_2]})-H(\lfloor h_{2}X_{2,D}^n\rfloor)=0
\end{equation*}
and for $n_E< n_2$ the second term contains more bits, and we can therefore use the chain rule and lemma \ref{Lemma_entropy_floorfunc_equal_bits} and show that
\begin{IEEEeqnarray}{rCl}\label{reminder_x2_helper_LT_converse}
&&H((\lfloor h_{E} X_{2,D}^n\rfloor )_{[:n_2]})-H(\lfloor h_{2}X_{2,D}^n\rfloor)\IEEEnonumber\\
&&= -H((\mathbf{\tilde{X}}_2^n)_{[n_E+1:]}|(\mathbf{\tilde{X}}_2^n)_{[:n_{E}]}).
\end{IEEEeqnarray}
We now split
the received signals in common and private parts. We start by adding two of the terms and split them apart
\begin{IEEEeqnarray*}{rCl}
&&2(H(Y_{1,D}^n)-H((Y_{2,D}^n)_{[:n_2]})) \\
&\leq & 2H((Y_{1,D}^n)_{[n_c+1:]})+ 2H((Y_{1,D}^n)_{[:n_c]})-2H((Y_{2,D}^n)_{[:n_2]}).
\end{IEEEeqnarray*}
Note that the private part $H((Y_{1,D}^n)_{[n_c+1:]})$ is zero for $n_{1}\leq n_{E}$. Now, counting from top to bottom, for $n_1\geq n_2$, $X_{1,D}^n$ has $n_c$ bit-levels in $(Y_{1,D})_{[:n_c]}$, while $X_{2,D}^n$ has $\eta:=\min\{n_E,\min\{n_{1},n_{2}\}\}=n_c-n_\Delta$ bit-levels. Therefore, $\eta$ represents the amount of bit-levels of the weaker signal in the common received signal part. Hence, for $n_2>n_1$, $X_{1,D}^n$ and $X_{2,D}^n$ have $\eta$ and $n_c$ bit-levels in that term, respectively. We need to account for this switch of indexing in the next part, where we analyse the entropy difference. We will use a method inspired by \cite{Fritschek2014} to show the following (for $n_1 \geq n_2$)
\begin{IEEEeqnarray*}{rCl}
&&2(H(Y_{1,D}^n)-H((Y_{2,D}^n)_{[:n_2]}) \\
&\leq & 2H((Y_{1,D}^n)_{[n_c+1:]})+ 2H((Y_{1,D}^n)_{[:n_c]})\\
&&-\:H((Y_{2,D}^n)_{[:n_2]}|\mathbf{\tilde{X}}_1^n)-H((Y_{2,D}^n)_{[:n_2]}|\mathbf{\tilde{X}}_2^n)\\
&= & 2H((Y_{1,D}^n)_{[n_c+1:]})+ 2H((Y_{1,D}^n)_{[:n_c]})\\
&&\:-H((\mathbf{\tilde{X}}_{2}^n)_{[:n_2]})-H((\mathbf{\tilde{X}}_{1}^n)_{[:n_2]})\\
&= & 2H((Y_{1,D}^n)_{[n_c+1:]})+ H((Y_{1,D}^n)_{[:n_c]})-H((\mathbf{\tilde{X}}_{2}^n)_{[:n_2]})\\
&&\:+H((f(\lfloor h_{1}X_{1,\text{D}}^n\rfloor ,\lfloor h_{2}X_{2,\text{D}}^n \rfloor) )_{[:n_c]})-H((\mathbf{\tilde{X}}_{1}^n)_{[:n_2]})\\
&\leq & 2H((Y_{1,D}^n)_{[n_c+1:]})+ H((Y_{1,D}^n)_{[:n_c]})-H((\mathbf{\tilde{X}}_{2}^n)_{[:n_2]})\\
&&\:+H( (\lfloor h_{2}X_{2,\text{D}}^n \rfloor)_{[:\eta]})+ H((\lfloor h_1X_{1,\text{D}}^n\rfloor)_{[:n_c]})-H((\mathbf{\tilde{X}}_{1}^n)_{[:n_2]})\\ 
&=& 2H((Y_{1,D}^n)_{[n_c+1:]})+ H((Y_{1,D}^n)_{[:n_c]})-H((\mathbf{\tilde{X}}_{2}^n)_{[:n_2]})\\
&&\:+H( (\mathbf{\tilde{X}}_{2}^n)_{[:\eta]})+ H((\mathbf{\tilde{X}}_{1}^n)_{[:n_c]})-H((\mathbf{\tilde{X}}_{1}^n)_{[:n_2]}),
\end{IEEEeqnarray*}

We now have for $n_1\geq n_2$ that 
\begin{equation*}
H((\mathbf{\tilde{X}}_{1}^n)_{[:n_c]})-H((\mathbf{\tilde{X}}_{1}^n)_{[:n_2]})\leq n(n_c-\min\{n_2,n_E\})^+\leq nn_\Delta,
\end{equation*} and 
\begin{equation*}
H( (\mathbf{\tilde{X}}_{2}^n)_{[:\eta]})-H((\mathbf{\tilde{X}}_{2}^n)_{[:n_2]})\leq n(\eta-\min \{n_2,n_E\})^+=0.
\end{equation*}
And for $n_2>n_1$ we get
\begin{equation*}
H((\mathbf{\tilde{X}}_{1}^n)_{[:\eta]})-H((\mathbf{\tilde{X}}_{1}^n)_{[:n_2]})\leq n(\eta-\min \{n_2,n_E\})^+=0,
\end{equation*} and 
\begin{equation*}
H( (\mathbf{\tilde{X}}_{2}^n)_{[:n_c]})-H((\mathbf{\tilde{X}}_{2}^n)_{[:n_2]})\leq n(n_c-\min\{n_2,n_E\})^+.
\end{equation*}
We remark that the last term gets $(n_2-n_E)^+$ for $n_1<n_E<n_2$, in which case we can use \eqref{reminder_x2_helper_LT_converse}, which has a length of $(n_2-n_E)$ bit-levels. Also for $n_E<n_1<n_2$ we have that $n(n_c-\min\{n_2,n_E\})^+=nn_\Delta$, by using \eqref{reminder_x2_helper_LT_converse} again, we see that for $n_2>n_1$
\begin{equation*}
H( (\mathbf{\tilde{X}}_{2}^n)_{[:n_c]})-H((\mathbf{\tilde{X}}_{2}^n)_{[:n_2]})\leq n(n_c-\min\{n_2,n_E\})^+=0.
\end{equation*}
We therefore have an additional term of $nn_\Delta$ for $n_1\geq n_2$.
Now one can divide all terms by two, resulting in
\begin{IEEEeqnarray*}{rCl}
&&H(Y_{1,D}^n)-H((Y_{2,D}^n)_{[:n_2]}) \\
&\leq & H((Y_{1,D}^n)_{[n_c+1:]})+ \frac{1}{2}H((Y_{1,D}^n)_{[:n_c]})+\frac{n}{2}(n_1-n_2)^+.
\end{IEEEeqnarray*}
Plugging all the results into the first equation yields 
\begin{IEEEeqnarray*}{rCl}
n(R-\epsilon)&\leq& n(n_p+\tfrac{1}{2}n_c+\tfrac{1}{2}(n_1-n_2)^++c).
\end{IEEEeqnarray*}
dividing by n and letting $n\rightarrow \infty$ shows the result.

For the case that $n_2>2n_1$ we have that
\begin{IEEEeqnarray*}{rCl}
n(R-\epsilon) &\leq & H(Y_{1,D}^n)-H(Y_{2,D}^n)+H(\lfloor h_{E} X_{2,D}^n\rfloor )\\
&&-\:H(\lfloor h_{2}X_{2,D}^n\rfloor )+nc\\
& \leq & H(\lfloor h_{1} X_{1,D}^n\rfloor)+H(\lfloor h_{2} X_{2,D}^n\rfloor )-H(Y_{2,D}^n|\mathbf{\tilde{X}}_1^n)\\
&&-\:H(\lfloor h_{2}X_{2,D}^n\rfloor )+H(\lfloor h_{E} X_{2,D}^n\rfloor )+nc\\
&=& H(\lfloor h_{1} X_{1,D}^n\rfloor)\leq n n_{1}
\end{IEEEeqnarray*}
and for the case that $3n_2<2n_1$ we have that
\begin{IEEEeqnarray*}{rCl}
nR &\leq & I(W;Y_{1,D}^n)-I(W;Y_{2,D}^n)+n(\epsilon+c)\\ 
&\leq & I(W;Y_{1,D}^n)-I(W;(Y_{2,D}^n)_{[:n_1-n_2]})+n(\epsilon+c)\\
&\leq &  H(Y_{1,D}^n)-H((Y_{2,D}^n)_{[:n_1-n_2]})\\
&&+\:H((\lfloor h_{E} X_{2,D}^n\rfloor )_{[:n_1-n_2]})-H(\lfloor h_{2} X_{2,D}^n\rfloor)+n(\epsilon+c)\\
&\leq &  H((Y_{1,D}^n)_{[:(n_1-n_2)]})-H((Y_{2,D}^n)_{[:n_1-n_2]}| \mathbf{\tilde{X}}_2^n)\\
&&+\:H((Y_{1,D}^n)_{[(n_1-n_2)+1:]}|(Y_{1,D}^n)_{[:(n_1-n_2)]})\\
&&+\:H((\lfloor h_{E} X_{2,D}^n\rfloor )_{[:n_1-n_2]})-H(\lfloor h_{2} X_{2,D}^n\rfloor)+n(\epsilon+c).
\end{IEEEeqnarray*}
One can show that
 \begin{IEEEeqnarray*}{rCl}
&&H((Y_{1,D}^n)_{[:(n_1-n_2)]})-H((Y_{2,D}^n)_{[:n_1-n_2]}| \mathbf{\tilde{X}}_2^n) \\
&\leq & n(n_1-n_2-n_E)^+
\end{IEEEeqnarray*}
and \begin{IEEEeqnarray*}{rCl}
&&H((\lfloor h_{E} X_{2,D}^n\rfloor )_{[:n_1-n_2]})-H(\lfloor h_{2} X_{2,D}^n\rfloor)\\
&\leq & n(\min\{n_1-n_2,n_E\}-n_2)\\
&=& n[n_E-n_2-(n_E-n_1+n_2)^+]^+
\end{IEEEeqnarray*}
and $H((Y_{1,D}^n)_{[(n_1-n_2)+1:]}|(Y_{1,D}^n)_{[:(n_1-n_2)]}) \leq nn_{2}$ which yields
\begin{IEEEeqnarray*}{rCl}
nR & \leq & nn_2+n(n_1-n_2-n_E)^+\\
&& +\:n[n_E-n_2-(n_E-n_1+n_2)^+]^++n(\epsilon+c)
\end{IEEEeqnarray*}

dividing by n and letting $n\rightarrow \infty$ shows the result.
\qed

%

\section{The Gaussian Multiple-Access Wiretap Channel}

In this section we analyse the Gaussian MAC-WT channel. 
As in the case for the WT channel with a Helper, we want to stick to the ideas of the corresponding linear deterministic model. This means we want to transfer the alignment and jamming structure to its Gaussian equivalent with layered lattice codes. This will lead to an achievable rate which is directly based on the deterministic rate. Moreover, we will make use of the previously developed ideas to convert the converse proof of the linear deterministic model, to the truncated model and therefore to the Gaussian model.

\subsection{Proof of Theorem 7: Achievable Sum-Rate for the G-MAC-WT}

We use the same framework as for the wiretap channel with a helper in section \ref{WT_Helper_Achievable_Scheme}. We partition the signal-to-noise ratio into intervals with power $\theta_{l}$, see eq. \eqref{theta_k_fractions}, where $l$ indicates the level. Each of these intervals plays the role of an $n_\Delta-$Interval of bit-levels in the linear deterministic scheme. We therefore partition the received power-to-noise ratio at $Y_1$ into intervals $\text{SNR}_1^{(1-\beta_1)}$. The users decompose the signals $X_i$ into a sum of independent sub-signals $X_i=\sum_{l=1}^{l_{max}}X_{il}$. And each signal uses the layered lattice codes as defined in section \ref{WT_Helper_Achievable_Scheme}.
Note that, w.l.o.g we look at the case $\text{SNR}_1>\text{SNR}_2$, which is $\beta_1<1$. Due to the symmetry of the users the case $\beta_1\geq 1$ follows immediately by interchanging both signals.
 As in the deterministic case, we have a private and common part. The common part is defined as the bit-levels $n_c:=\min\{n_E+n_\Delta,n_1\}$. The part $n_E+n_\Delta$ corresponds to $\text{SNR}_1^{\beta_2+(1-\beta_1)}$ in the Gaussian model. The opposing remainder is therefore $\text{SNR}_1^{\beta_1-\beta_2}$ and we get the common power-to-noise ratio as 
\begin{equation}
\text{SNR}_c:=\text{SNR}_1-\max\{1,\text{SNR}_1^{\beta_1-\beta_2}\},
\end{equation}
 while the private part has a power-to-noise ratio of 
 \begin{equation}
\text{SNR}_p:=\max\{1,\text{SNR}_1^{\beta_1-\beta_2}\}-1,
 \end{equation} 
 exactly as in the case of the wiretap channel with a helper. 
However, due to the modified scheme where both users jam and align their jamming signals at the legitimate receiver (see section \ref{Ach_Scheme_LD_MAC_WT}) we have a different number of used levels for messaging. We have
\begin{equation}
l_{\text{u}}:= 2\left\lfloor \frac{\min\{1+\beta_2-\beta_1,1\}}{3(1-\beta_1)}\right\rfloor
\end{equation} used levels for messaging, where
each one supports a rate of $r_l\geq \tfrac{1}{2}\log \text{SNR}_1^{(1-\beta_1)}-\tfrac{1}{2}$. And we therefore have a sum rate of 
\begin{equation}
r^c=l_{\text{u}}(\tfrac{1}{2}\log \text{SNR}_1^{(1-\beta_1)}-\tfrac{1}{2}),
\label{Gauss_rate_MAC_WT}
\end{equation} for the whole common alignment part.
Moreover, we need to consider the remainder term, which is allocated between the alignment structure and the noise floor or the private part. We see from the deterministic scheme that for $1-(\tfrac{3}{2}l_\text{u}+1)(1-\beta_1) < \min\{ \beta_1-\beta_2,0\}$ we can achieve a rate of
\begin{equation*}
r^R \geq \tfrac{1}{2}\log \text{SNR}_1^{1-\tfrac{3}{2}l_\text{u}(1-\beta_1)} - \tfrac{1}{2}\log \text{SNR}_1^{\min\{ \beta_1-\beta_2,0\}}-\tfrac{1}{2}.
\end{equation*} Moreover, for $1-(\tfrac{3}{2}l_\text{u}+2)(1-\beta_1) < \min\{ \beta_1-\beta_2,0\}\leq 1-(\tfrac{3}{2}l_\text{u}+1)(1-\beta_1) $
we have
\begin{equation*}
r^R \geq \tfrac{1}{2}\log \text{SNR}_1^{(1-\beta_1)}-\tfrac{1}{2},
\end{equation*} and for $\min\{ \beta_1-\beta_2,0\} \leq 1-(\tfrac{3}{2}l_\text{u}+2)(1-\beta_1)$ we have 
\begin{IEEEeqnarray*}{rCl}
r^R &\geq& \tfrac{1}{2}\log \text{SNR}_1^{(1-\beta_1)}+\tfrac{1}{2}\log \text{SNR}_1^{1-\tfrac{3}{2}l_\text{u}(1-\beta_1)}\\
&&\: - \tfrac{1}{2}\log \text{SNR}_1^{\min\{ \beta_1-\beta_2,0\}}-1.
\end{IEEEeqnarray*}
We therefore get a total rate of 
$
R_{\Sigma}=r^p+r^c+r^R.$
The case for $\text{SNR}_2> \text{SNR}_1$ can be shown similarly.
\qed

\subsection{Proof of Theorem 8: Sum-Rate Bound for the G-MAC-WT}
We use a similar approach as for the Gaussian WT with a helper, with the same framework developed in section \ref{Developing a Converse from LD-Bounds}. This means we also use the truncated deterministic model 
\begin{IEEEeqnarray}{rCl}\label{Truncated_Model2}
Y_{1,\text{D}}&=&\lfloor h_{1}X_{1,\text{D}}\rfloor +\lfloor h_{2}X_{2,\text{D}} \rfloor\IEEEyessubnumber\\
Y_{2,\text{D}}&=&\lfloor h_{E}X_{2,\text{D}} \rfloor + \lfloor h_{E}X_{1,\text{D}}\rfloor \IEEEyessubnumber,
\end{IEEEeqnarray} 
which can be shown to be within a constant gap to the Gaussian channel, see \eqref{Constant_Gap}.

\begin{proof} We begin with the following derivations
\begin{IEEEeqnarray}{rCl}
 \IEEEeqnarraymulticol{3}{l}{
n(R_{\Sigma}-\epsilon)}\\*\quad 
&=& I(W_1,W_2;Y_{1,\text{G}}^n)-I(W_1,W_2;Y_{2,\text{G}}^n)\IEEEnonumber\\
&\leq & I(W_1,W_2;Y_{1,\text{D}}^n)-I(W_1,W_2;Y_{2,\text{D}}^n)+nc\IEEEnonumber\\
&\leq & I(W_1,W_2;Y_{1,\text{D}}^n,Y_{2,\text{D}}^n)-I(W_1,W_2;Y_{2,\text{D}}^n)+nc\IEEEnonumber\\
&\leq & I(W_1,W_2;Y_{1,\text{D}}|Y_{2,\text{D}})+nc\IEEEnonumber\\
&\leq & I(X_{1,\text{D}}^n,X_{2,\text{D}}^n;Y_{1,\text{D}}^n|Y_{2,\text{D}}^n)+nc\IEEEnonumber\\
&=&H(Y_{1,\text{D}}^n|Y_{2,\text{D}}^n)
	-H(Y_{1,\text{D}}^n|Y_{2,\text{D}}^n,X_{1,\text{D}}^n,X_{2,\text{D}}^n)+nc\IEEEnonumber\\
&\overset{(b)}{=}&H(Y_{1,\text{D}}^n|Y_{2,\text{D}}^n)+nc\IEEEnonumber\\
&\overset{(c)}{\leq} & H(Y_{1,\text{D},c}^n|Y_{2,\text{D}})+H(Y_{1,\text{D},p}^n|Y_{2,\text{D}},Y_{1,\text{D},c})+nc
\label{UB:EQ1*}
\end{IEEEeqnarray}
 where we used basic techniques such as Fano's inequality and the chain rule. Step $(a)$ introduces the secrecy constraint \eqref{security_constraint}, while we used  the chain rule, non-negativity of mutual information and the data processing inequality in the following lines. Step $(b)$ follows from the fact that $Y_{1,\text{D}}^n$ is a function of $(X_{1,\text{D}}^n,X_{2,\text{D}}^n)$. Note that due to the definition of the common and the private part\footnote{The common part is defined as $Y_{1,\text{D},c}^n=(Y_{1,\text{D}}^n)_{[:n_c]}$, and the private part as $Y_{1,\text{D},p}^n=(Y_{1,\text{D}}^n)_{[n_c+1:]}$.} of $Y_{1,\text{D}}^n$, it follows that $H(Y_{1,\text{D},p}^n|Y_{2,\text{D}},Y_{1,\text{D},c})=0$ for $n_E\geq n_2$. For step $(c)$, we used lemma \ref{Lemma_floorsplitting}, the data-processing inequality and the chain-rule. We now extend the strategy of \cite{XieUlukusOneHop}, of bounding a single signal part, to asymmetrical channel gains
\begin{IEEEeqnarray}{rCl}
 \IEEEeqnarraymulticol{3}{l}{
n(R_1- \epsilon_3)}\\
&\leq & I(X_{1,\text{D}}^n;Y_{1,\text{D}}^n)\IEEEnonumber\\
&\leq & I(X_{1,\text{D}}^n;Y_{1,\text{D},c}^n)+I(X_{1,\text{D}}^n;Y_{1,\text{D},p}^n|Y_{1,\text{D},c}^n)\IEEEnonumber\\
&=& H(Y_{1,\text{D},c}^n)-H(Y_{1,\text{D},c}^n|X_{1,\text{D}}^n)+I(X_{1,\text{D}}^n;Y_{1,\text{D},p}^n|Y_{1,\text{D},c}^n)\IEEEnonumber\\
&=&H(Y_{1,\text{D},c}^n)-H((\lfloor h_2X_{2,\text{D}}^n \rfloor)_{[:n_c]})+I(X_{1,\text{D}}^n;Y_{1,\text{D},p}^n|Y_{1,\text{D},c}^n)\IEEEnonumber
\label{UB:EQ3*}
\end{IEEEeqnarray}
and it therefore holds that
\begin{IEEEeqnarray}{rCl}
H((\lfloor h_2X_{2,\text{D}}^n \rfloor)_{[:n_c]}) &\leq & H(Y_{1,\text{D},c}^n)\\
&&\:+I(X_{1,\text{D}}^n;Y_{1,\text{D},p}^n|Y_{1,\text{D},c}^n)-n(R_1-\epsilon_3)\IEEEnonumber.
\label{UB:EQ4*}
\end{IEEEeqnarray}
The same can be shown for $H((\lfloor h_2X_{1,\text{D}}^n \rfloor)_{[:n_c]})$, where it holds that
\begin{IEEEeqnarray}{rCl}
H((\lfloor h_2X_{1,\text{D}}^n\rfloor)_{[:n_c]}) &\leq & H(Y_{1,\text{D},c}^n)\\
&&\:+I(X_{2,\text{D}}^n;Y_{1,\text{D},p}^n|Y_{1,\text{D},c}^n)-n(R_2-\epsilon_3).\IEEEnonumber
\label{UB:EQ5*}
\end{IEEEeqnarray}
Moreover, we have that
\begin{IEEEeqnarray*}{rCl}
 \IEEEeqnarraymulticol{3}{l}{
I(X_{1,\text{D}}^n;Y_{1,\text{D},p}^n|Y_{1,\text{D},c}^n)+I(X_{2,\text{D}}^n;Y_{1,\text{D},p}^n|Y_{1,\text{D},c}^n)}\\
&=&\:2H(Y_{1,\text{D},p}^n|Y_{1,\text{D},c}^n)-H(Y_{1,\text{D},p}^n|Y_{1,\text{D},c}^n,X_{1,\text{D}}^n)\\
&&\:-H(Y_{1,\text{D},p}^n|Y_{1,\text{D},c}^n,X_{2,\text{D}}^n)\\
&=& 2H(Y_{1,\text{D},p}^n|Y_{1,\text{D},c}^n)-\:H((\lfloor h_2 X_{2,\text{D}}^n\rfloor)_{[n_c+1:]}|Y_{1,\text{D},c}^n)\\
&&\:-H((\lfloor h_1 X_{1,\text{D}}^n\rfloor)_{[n_c+1:]}|Y_{1,\text{D},c}^n)\\
&= &\:H(Y_{1,\text{D},p}^n|Y_{1,\text{D},c}^n)\IEEEyesnumber.
\label{UB:EQ10*}
\end{IEEEeqnarray*}
The key idea for the various cases is now to bound the term $H(Y_{1,\text{D},c}^n|Y_{2,\text{D}})$, or equivalently $H(Y_{1,\text{D}}^n|Y_{2,\text{D}})$ for $n_E>n_2$, in an appropriate way, to be able to use \eqref{UB:EQ4*} and \eqref{UB:EQ5*} on \eqref{UB:EQ1*}.
We start with the first case:
\subsubsection{Case $n_2\geq n_E$}
Here we have a none vanishing private part, due to the definition of $Y_{1,\text{D},c}^n$ and therefore need to bound the term $H(Y_{1,\text{D},c}^n|Y_{2,\text{D}})$. Note that due to the definition of $Y_{1,\text{D},c}^n$ and the specific case, we have that $H((\lfloor h_2X_{2,\text{D}}^n \rfloor)_{[:n_c]})=H(\lfloor h_EX_{2,D}^n \rfloor)$. We look into the first term of equation \eqref{UB:EQ1*} and show that
\begin{IEEEeqnarray*}{rCl}
 \IEEEeqnarraymulticol{3}{l}{
H(Y_{1,\text{D},c}^n|Y_{2,\text{D}})}\\
&=& H(Y_{1,\text{D},c}^n,Y_{2,\text{D}})-H(Y_{2,\text{D}})\\
&\leq & H(Y_{1,\text{D},c}^n,\lfloor h_{E}X_{2,\text{D}} \rfloor , \lfloor h_{E}X_{1,\text{D}}\rfloor)-H(Y_{2,\text{D}})\\
&=& H(\lfloor h_{E}X_{2,\text{D}} \rfloor,\lfloor h_{E}X_{1,\text{D}}\rfloor)-H(Y_{2,\text{D}})\\
&&\:+H(Y_{1,\text{D},c}^n|\lfloor h_{E}X_{2,\text{D}} \rfloor,\lfloor h_{E}X_{1,\text{D}}\rfloor)\\
& \leq &H(\lfloor h_{E}X_{1,\text{D}} \rfloor)+H(\lfloor h_{E}X_{2,\text{D}} \rfloor)-H(Y_{2,\text{D}}|X_{2,\text{D}})\\
&&+\:H(Y_{1,\text{D},c}^n|\lfloor h_{E}X_{2,\text{D}} \rfloor,\lfloor h_{E}X_{1,\text{D}}\rfloor)\\
&=& H(\lfloor h_{E}X_{2,\text{D}} \rfloor)+H(Y_{1,\text{D},c}^n|\lfloor h_{E}X_{2,\text{D}} \rfloor,\lfloor h_{E}X_{1,\text{D}}\rfloor).\IEEEyesnumber
\label{UB:EQ2*}
\end{IEEEeqnarray*}
Observe that the second term of equation \eqref{UB:EQ2*} is dependent on the specific regime. We can bound this term by
\begin{equation}
H(Y_{1,\text{D},c}^n|\lfloor h_{E}X_{2,\text{D}} \rfloor,\lfloor h_{E}X_{1,\text{D}}\rfloor)\leq n(n_c-n_E)= nn_\Delta.
\end{equation}
Note that the choice of $\lfloor h_{E}X_{2,\text{D}} \rfloor$ in \eqref{UB:EQ2*} as remaining signal part was arbitrary due to our assumption that both signals $\lfloor h_{E}X_{1,\text{D}} \rfloor$ and $\lfloor h_{E}X_{2,\text{D}} \rfloor$ have the same signal strength. Moreover, it follows on the same lines that
\begin{equation}
H(Y_{1,\text{D},c}^n|Y_{2,\text{D}})\leq H(\lfloor h_{E}X_{1,\text{D}} \rfloor)+nn_\Delta.
\label{UB:EQ6*}
\end{equation}
Looking at this result, its intuitive that one can also show the stronger result
\begin{equation}
H(Y_{1,\text{D},c}^n|Y_{2,\text{D}})\leq H((\lfloor h_{1}X_{1,\text{D}}\rfloor)_{[:n_c]})
\label{UB:EQ7*}
\end{equation}
for the case that $n_2 \geq n_E$.
This can be shown by considering a similar strategy as in \eqref{UB:EQ2*} 
\begin{IEEEeqnarray*}{rCl}
\IEEEeqnarraymulticol{3}{l}{
H(Y_{1,\text{D},c}^n|Y_{2,\text{D}})}\\* \quad 
&=& H(Y_{1,\text{D},c}^n,Y_{2,\text{D}})-H(Y_{2,\text{D}})\\
&\leq & H(Y_{2,\text{D}},(\lfloor h_{1}X_{1,\text{D}}\rfloor)_{[:n_c]},(\lfloor h_{2}X_{2,\text{D}}\rfloor)_{[:n_E]})-H(Y_{2,\text{D}})\\
&=& H((\lfloor h_{1}X_{1,\text{D}}\rfloor)_{[:n_c]},(\lfloor h_{2}X_{2,\text{D}}\rfloor)_{[:n_E]})-H(Y_{2,\text{D}})\\
&&\:+H(Y_{2,\text{D}}|(\lfloor h_{1}X_{1,\text{D}}\rfloor)_{[:n_c]},(\lfloor h_{2}X_{2,\text{D}}\rfloor)_{[:n_E]})\\
& \leq & H((\lfloor h_{1}X_{1,\text{D}}\rfloor)_{[:n_c]})+H((\lfloor h_{2}X_{2,\text{D}}\rfloor)_{[:n_E]})\\
&&\:-H(Y_{2,\text{D}}|X_{1,\text{D}})\\
&&+\:H(Y_{2,\text{D}}|(\lfloor h_{1}X_{1,\text{D}}\rfloor)_{[:n_c]},(\lfloor h_{2}X_{2,\text{D}}\rfloor)_{[:n_E]}))\\
&=& H((\lfloor h_{1}X_{1,\text{D}}\rfloor)_{[:n_c]})\\
&&\:+H(Y_{2,\text{D}}|(\lfloor h_{1}X_{1,\text{D}}\rfloor)_{[:n_c]},(\lfloor h_{2}X_{2,\text{D}}\rfloor)_{[:n_E]})),\IEEEyesnumber
\label{UB:EQ8*}
\end{IEEEeqnarray*} where 
\begin{IEEEeqnarray*}{rCl}
&&H(Y_{2,\text{D}}|(\lfloor h_{1}X_{1,\text{D}}\rfloor)_{[:n_c]},(\lfloor h_{2}X_{2,\text{D}}\rfloor)_{[:n_E]}))\IEEEyesnumber\\
&&\leq n(n_E-n_2)^+=0.
\label{UB:EQ9*}
\end{IEEEeqnarray*}
We combine one sum-rate inequality \eqref{UB:EQ1*} with \eqref{UB:EQ2*} and one with \eqref{UB:EQ8*}.
Moreover, we plug \eqref{UB:EQ4*} and \eqref{UB:EQ5*} into the corresponding bound, which yields
\begin{IEEEeqnarray*}{rCl}
n(2R_1+R_2 -\epsilon_6) &\leq & H(Y_{1,\text{D},c}^n)+I(X_{2,\text{D}}^n;Y_{1,\text{D},p}^n|Y_{1,\text{D},c}^n)\\
&&+\:H(Y_{1,\text{D},p}^n|Y_{2,\text{D}},Y_{1,\text{D},c})
\end{IEEEeqnarray*} and
\begin{IEEEeqnarray*}{rCl}
n(R_1+2R_2 -\epsilon_7) &\leq & H(Y_{1,\text{D},c}^n)+I(X_{1,\text{D}}^n;Y_{1,\text{D},p}^n|Y_{1,\text{D},c}^n)\\
&&+\:H(Y_{1,\text{D},p}^n|Y_{2,\text{D}},Y_{1,\text{D},c})+nn_\Delta.
\end{IEEEeqnarray*}
A summation of these results gives
 \begin{IEEEeqnarray*}{rCl}
 \IEEEeqnarraymulticol{3}{l}{
3n(R_1+R_2)-n\epsilon_8}\\* \quad
 &\leq & 2H(Y_{1,\text{D},c}^n)+I(X_{1,\text{D}}^n;Y_{1,\text{D},p}^n|Y_{1,\text{D},c}^n)\\
&&+\:I(X_{2,\text{D}}^n;Y_{1,\text{D},p}^n|Y_{1,\text{D},c}^n)+nn_\Delta\\
&&+\:2H(Y_{1,\text{D},p}^n|Y_{2,\text{D}},Y_{1,\text{D},c})\\
&=&2H(Y_{1,\text{D},c}^n)+H(Y_{1,\text{D},p}^n|Y_{1,\text{D},c}^n)+nn_\Delta\\
&&+\:2H(Y_{1,\text{D},p}^n|Y_{2,\text{D}},Y_{1,\text{D},c}),
\end{IEEEeqnarray*}
where we used \eqref{UB:EQ10*}. Now, because $H(Y_{1,\text{D},p}^n|Y_{2,\text{D}},Y_{1,\text{D},c})\leq nn_p$, $H(Y_{1,\text{D},p}^n|Y_{1,\text{D},c}^n)\leq nn_p$ and $H(Y_{1,\text{D},c}^n)\leq n_c$, we have 
\begin{IEEEeqnarray*}{rCl}
3n(R_1+R_2) -n\epsilon_8 &\leq & 2nn_c+3nn_p+nn_\Delta.
\end{IEEEeqnarray*}
Dividing by $3n$ and letting $n\rightarrow\infty$ shows the result.
\subsubsection{Case $n_E>n_2$}
First, we assume that $n_E \geq n_1$, and include a short proof for $n_1>n_E\geq n_2$ at the end of this subsection. For this case, the private part $Y_{1,\text{D},p}^n$ is zero, due to the definition of the private part and $n_E>n_2$. It follows that \eqref{UB:EQ1*} is
\begin{equation}
n(R_1+R_2)\leq H(Y_{1,\text{D}}^n|Y_{2,\text{D}}^n).
\label{UB:EQ11*}
\end{equation}
Moreover, we have that 
\begin{equation*}
H(\lfloor h_{2}X_{2,\text{D}}\rfloor)=H((\lfloor h_{2}X_{2,\text{D}}\rfloor)_{[:n_E]})\leq H(\lfloor h_{E}X_{2,\text{D}} \rfloor),
\end{equation*} which is why we need to bound \eqref{UB:EQ11*} by
 $H(\lfloor h_{1}X_{1,\text{D}}\rfloor)$ and $H(\lfloor h_{2}X_{2,\text{D}}\rfloor)$. We therefore modify \eqref{UB:EQ8*} to fit our case in the following way
\begin{IEEEeqnarray*}{rCl}
\IEEEeqnarraymulticol{3}{l}{
H(Y_{1,\text{D}}^n|Y_{2,\text{D}})}\\* \quad 
&=& H(Y_{1,\text{D}}^n,Y_{2,\text{D}})-H(Y_{2,\text{D}})\\
&\leq & H(Y_{2,\text{D}},\lfloor h_{1}X_{1,\text{D}}\rfloor,\lfloor h_{2}X_{2,\text{D}}\rfloor)-H(Y_{2,\text{D}})\\
&=& H(\lfloor h_{1}X_{1,\text{D}}\rfloor,\lfloor h_{2}X_{2,\text{D}}\rfloor)-H(Y_{2,\text{D}})\\
&&\:+H(Y_{2,\text{D}}|\lfloor h_{1}X_{1,\text{D}}\rfloor,\lfloor h_{2}X_{2,\text{D}}\rfloor)\\
&=& H(\lfloor h_{1}X_{1,\text{D}}\rfloor,\lfloor h_{2}X_{2,\text{D}}\rfloor)+H(Y_{2,\text{D},c}^n|\lfloor h_{1}X_{1,\text{D}}\rfloor,\lfloor h_{2}X_{2,\text{D}}\rfloor)\\
&&+\:H(Y_{2,\text{D},p}^n|\lfloor h_{1}X_{1,\text{D}}\rfloor,\lfloor h_{2}X_{2,\text{D}}\rfloor, Y_{2,\text{D},c}^n)\\
&&-\:H(Y_{2,\text{D},c}^n)-H(Y_{2,\text{D},p}^n|Y_{2,\text{D},c}^n)\\
&\leq &H(\lfloor h_{1}X_{1,\text{D}}\rfloor,\lfloor h_{2}X_{2,\text{D}}\rfloor)\\
&&\:+H(Y_{2,\text{D},c}^n|\lfloor h_{1}X_{1,\text{D}}\rfloor,\lfloor h_{2}X_{2,\text{D}}\rfloor)-H(Y_{2,\text{D},c}^n),
\label{UB:EQ12*}
\end{IEEEeqnarray*}
where $Y_{2,\text{D},c}^n=(Y_{2,\text{D}}^n)_{[:n_1]}$ and $Y_{2,\text{D},p}^n=(Y_{2,\text{D}}^n)_{[n_1+1:]}$.
Now, we can show that
\begin{IEEEeqnarray*}{rCl}
\IEEEeqnarraymulticol{3}{l}{
H(Y_{1,\text{D}}^n|Y_{2,\text{D}})}\\* \quad 
&\leq &H(\lfloor h_{1}X_{1,\text{D}}\rfloor,\lfloor h_{2}X_{2,\text{D}}\rfloor)\\
&&\:+H(Y_{2,\text{D},c}^n|\lfloor h_{1}X_{1,\text{D}}\rfloor,\lfloor h_{2}X_{2,\text{D}}\rfloor)-H(Y_{2,\text{D},c}^n),\\
&=&H(\lfloor h_{1}X_{1,\text{D}}\rfloor,\lfloor h_{2}X_{2,\text{D}}\rfloor)\\
&&\:-H(Y_{2,\text{D},c}^n)+H((\lfloor h_{E}X_{2,\text{D}}\rfloor)_{[:n_1]}|\lfloor h_{2}X_{2,\text{D}}\rfloor)\\
&\leq &H(\lfloor h_{1}X_{1,\text{D}}\rfloor)+H(\lfloor h_{2}X_{2,\text{D}}\rfloor)-H(Y_{2,\text{D},c}^n|X_{2,\text{D}})\\
&&+\:H((\lfloor h_{E}X_{2,\text{D}}\rfloor)_{[:n_1]}|\lfloor h_{2}X_{2,\text{D}}\rfloor)\\
&=& H(\lfloor h_{2}X_{2,\text{D}}\rfloor)+H((\lfloor h_{E}X_{2,\text{D}}\rfloor)_{[:n_1]}|\lfloor h_{2}X_{2,\text{D}}\rfloor)\\
&\leq &H(\lfloor h_{2}X_{2,\text{D}}\rfloor)+ nn_\Delta,\IEEEyesnumber
\label{UB:EQ13*}
\end{IEEEeqnarray*}where the last inequality follows because we have that 
\begin{IEEEeqnarray*}{rCl}
&&H((\lfloor h_{E}X_{2,\text{D}}\rfloor)_{[:n_1]}|\lfloor h_{2}X_{2,\text{D}}\rfloor)\\
&&\:=H((\lfloor h_{E}X_{2,\text{D}}\rfloor)_{[n_2+1:n_1]}|\lfloor h_{2}X_{2,\text{D}}\rfloor),\IEEEyesnumber
\label{UB:EQ15*}
\end{IEEEeqnarray*}
due to lemma \ref{Lemma_entropy_floorfunc_equal_bits} and the chain-rule.
Bounding $H(Y_{1,\text{D}}^n|Y_{2,\text{D}})$ by $H(\lfloor h_{1}X_{1,\text{D}}\rfloor)$ requires more work. We have a redundancy in the negative entropy terms, with which we can cancel the $H((\lfloor h_{E}X_{2,\text{D}}\rfloor)_{[:n_1]}|\lfloor h_{2}X_{2,\text{D}}\rfloor)$ term in the following way
\begin{IEEEeqnarray*}{rCl}
\IEEEeqnarraymulticol{3}{l}{
H(Y_{1,\text{D}}^n|Y_{2,\text{D}})}\\* \quad 
&\leq &H(\lfloor h_{1}X_{1,\text{D}}\rfloor,\lfloor h_{2}X_{2,\text{D}}\rfloor)\\
&&\:+H(Y_{2,\text{D},c}^n|\lfloor h_{1}X_{1,\text{D}}\rfloor,\lfloor h_{2}X_{2,\text{D}}\rfloor)-H(Y_{2,\text{D},c}^n),\\
&\leq &H(\lfloor h_{1}X_{1,\text{D}}\rfloor)+H(\lfloor h_{2}X_{2,\text{D}}\rfloor)-H((Y_{2,\text{D},c}^n)_{[:n_2]}|X_{1,\text{D}})\\
&&+\:H((\lfloor h_{E}X_{2,\text{D}}\rfloor)_{[:n_1]}|\lfloor h_{2}X_{2,\text{D}}\rfloor)\\
&&\:-H((Y_{2,\text{D},c}^n)_{[n_2+1:]}|X_{1,\text{D}},(Y_{2,\text{D},c}^n)_{[:n_2]})\\
&=& H(\lfloor h_{1}X_{1,\text{D}}\rfloor)-H((Y_{2,\text{D},c}^n)_{[n_2+1:]}|X_{1,\text{D}},(Y_{2,\text{D},c}^n)_{[:n_2]})\\
&&+\:H((\lfloor h_{E}X_{2,\text{D}}\rfloor)_{[:n_1]}|\lfloor h_{2}X_{2,\text{D}}\rfloor)\\
&\leq &H(\lfloor h_{1}X_{1,\text{D}}\rfloor)+H((\lfloor h_{E}X_{2,\text{D}}\rfloor)_{[:n_1]}|\lfloor h_{2}X_{2,\text{D}}\rfloor)\\
&&-\:H((Y_{2,\text{D},c}^n)_{[n_2+1:]}|X_{1,\text{D}},(Y_{2,\text{D},c}^n)_{[:n_2]},\lfloor h_{2}X_{2,\text{D}}\rfloor)\\
&=&H(\lfloor h_{1}X_{1,\text{D}}\rfloor)-H((Y_{2,\text{D},c}^n)_{[n_2+1:]}|X_{1,\text{D}},\lfloor h_{2}X_{2,\text{D}}\rfloor)\\
&&+\:H((\lfloor h_{E}X_{2,\text{D}}\rfloor)_{[:n_1]}|\lfloor h_{2}X_{2,\text{D}}\rfloor)\\
&=& H(\lfloor h_{1}X_{1,\text{D}}\rfloor)-H((\lfloor h_{E}X_{2,\text{D}}\rfloor)_{[n_2+1:n_1]}|\lfloor h_{2}X_{2,\text{D}}\rfloor)\\
&&\:+H((\lfloor h_{E}X_{2,\text{D}}\rfloor)_{[:n_1]}|\lfloor h_{2}X_{2,\text{D}}\rfloor)\\
&= &H(\lfloor h_{1}X_{1,\text{D}}\rfloor),\IEEEyesnumber
\label{UB:EQ14*}
\end{IEEEeqnarray*}
where the last step follows due to equation \eqref{UB:EQ15*}.
Now we can bound one \eqref{UB:EQ11*} with \eqref{UB:EQ13*} and one with \eqref{UB:EQ14*}. Moreover, we use \eqref{UB:EQ4*} and \eqref{UB:EQ5*} on the result. Note that due to our regime, \eqref{UB:EQ4*} becomes
\begin{equation}
H(\lfloor h_{2}X_{2,\text{D}}\rfloor) \leq H(Y_{1,\text{D}}^n)-n(R_1+\epsilon_3),
\end{equation}
while \eqref{UB:EQ5*} becomes
\begin{equation}
H(\lfloor h_{1}X_{1,\text{D}}\rfloor) \leq H(Y_{1,\text{D}}^n)-n(R_2+\epsilon_4).
\end{equation}
Putting everything together results in
 \begin{IEEEeqnarray*}{rCl}
3n(R_1+R_2)-n\epsilon_8 &\leq & 2nn_c+nn_\Delta.
\end{IEEEeqnarray*}
Dividing by $3n$ and letting $n\rightarrow\infty$ shows the result.
We need to modify a bound on $H(Y_{1,\text{D}}^n|Y_{2,\text{D}})$, if the signal strength $n_E$ lies in between $n_1$ and $n_2$. In \eqref{UB:EQ13*}, we see that 
\begin{equation}
H(\lfloor h_{1}X_{1,\text{D}}\rfloor)-H(Y_{2,\text{D},c}^n|X_{2,\text{D}})\leq n(n_1-n_E)^+.
\end{equation}
Moreover, we have that 
\begin{equation}
H((\lfloor h_{E}X_{2,\text{D}}\rfloor)_{[:n_1]}|\lfloor h_{2}X_{2,\text{D}}\rfloor)\leq n(n_E-n_2)^+.
\end{equation}
Both changes cancel and we get the same result as \eqref{UB:EQ13*}. The result follows on the same lines as in the previous derivation.
\end{proof}


\section{Conclusions}
We have shown an achievable scheme for both the Gaussian multiple-access wiretap channel and the Gaussian wiretap channel with a helper.
We used the linear deterministic approximation of both models to gain insights into the structure and devised novel achievable schemes based on orthogonal bit-level alignment to achieve secrecy. These techniques can be summarized as signal-scale alignment methods, where we used jamming alignment at the eavesdropper in the signal-scale, while minimizing the negative effect at the legitimate receiver. Both results were then transferred to the Gaussian model, by utilizing layered lattice coding. Moreover, we developed converse proofs for both models, which achieve a constant-gap bound for certain signal power regimes. Those converse techniques were developed for the LD model and then transferred to a truncated deterministic model, which in turn is within a constant-gap of the integer-input integer-output model. The integer-input integer-output model yields converse proofs for the Gaussian models, by invoking a result of \cite{Mukherjee2017}. Since our results hold for asymmetrical channel gains and are dependent on those ratios, they give insights into the secure g.d.o.f. and converge to the known s.d.o.f. results for the channel gain ratio approaching one. 



\bibliographystyle{./IEEEtran}
\bibliography{./ref}

\end{document}